\documentclass[acmsmall]{acmart}
\acmJournal{PACMPL}
\acmVolume{1}
\acmNumber{OOPSLA} 
\acmArticle{1}
\acmYear{2020}
\acmMonth{1} 
\acmDOI{} 
\startPage{1}

\setcopyright{none}

\usepackage{hyperref}       
\usepackage{xurl}
\usepackage{booktabs}       
\usepackage{amsfonts}       
\usepackage{nicefrac}       
\usepackage{microtype}      
\usepackage{graphicx}
\usepackage{xspace}
\usepackage{amsmath, blkarray}
\usepackage{amsthm}
\usepackage{xpatch}
\usepackage{bm}
\usepackage{amssymb}
\usepackage{multirow}
\usepackage{algorithmic}
\usepackage[linesnumbered,ruled,vlined]{algorithm2e}
\usepackage{adjustbox}
\usepackage{listings,lstautogobble}
\usepackage{subcaption} 
\usepackage{tikz}
\usetikzlibrary{matrix}
\usepackage{xcolor}
\usepackage{flushend}
\usepackage{textcomp}
\usepackage{mathtools}
\usepackage{caption}
\usepackage{balance,url}
\usepackage{pifont}
\usepackage{pgfplots}
\pgfplotsset{
    compat=1.4,
    small,
    legend columns=2, 
    legend style={
        at={(0.99,0.99)},
        anchor=north east,
        font=\bfseries,
        /tikz/column 2/.style={
            column sep=5pt,
        },        
    },
    label style={font=\sffamily\small\bfseries}
}%

\theoremstyle{definition}

\newcommand{\etal}{\hbox{et al.}\xspace}
\newcommand{\eg}{\hbox{\emph{e.g.}}\xspace}
\newcommand{\ie}{\hbox{\emph{i.e.}}\xspace}
\newcommand{\st}{\hbox{\emph{s.t.}}\xspace}
\newcommand{\wrt}{\hbox{\emph{w.r.t.}}\xspace}
\newcommand{\etc}{\hbox{\emph{etc.}}\xspace}
\newcommand{\resp}{\hbox{\emph{resp.}}\xspace}

\newcommand{\ignore}[1]{\iffalse #1 \fi}

\newcommand{\tabincell}[2]{\begin{tabular}{@{}#1@{}}#2\end{tabular}}

\DeclareMathOperator*{\argmax}{argmax}

\makeatletter
\newcommand\Label[1]{&\refstepcounter{equation}(\theequation)\ltx@label{#1}&}
\makeatother

\makeatletter
\newenvironment{btHighlight}[1][]
{\begingroup\tikzset{bt@Highlight@par/.style={#1}}\begin{lrbox}{\@tempboxa}}
	{\end{lrbox}\bt@HL@box[bt@Highlight@par]{\@tempboxa}\endgroup}

\newcommand\btHL[1][]{%
	\begin{btHighlight}[#1]\bgroup\aftergroup\bt@HL@endenv%
	}
	\def\bt@HL@endenv{%
	\end{btHighlight}%
	\egroup
}
\newcommand{\bt@HL@box}[2][]{%
	\tikz[#1]{%
		\pgfpathrectangle{\pgfpoint{1pt}{0pt}}{\pgfpoint{\wd #2}{\ht #2}}%
		\pgfusepath{use as bounding box}%
		\node[anchor=base west, fill=orange!25,outer sep=.5pt,inner xsep=0.5pt, inner ysep=0.15pt, rounded corners=1pt, minimum height=\ht\strutbox-.1pt,#1]{\raisebox{.01pt}{\strut}\strut\usebox{#2}};
	}%
}
\makeatother

\definecolor{codegreen}{rgb}{0,0.6,0}
\definecolor{codegray}{rgb}{0.5,0.5,0.5}
\definecolor{codepurple}{rgb}{0.58,0,0.82}

\newcounter{lstannotation}
\lstdefinestyle{mystyle}{
      frame=single,
      framexleftmargin=0pt,
      commentstyle=\color{green},
      keywordstyle=\color{blue}\bfseries,
      numberstyle=\tiny\color{gray},
      stringstyle=\color{purple},
      basicstyle=\tiny\ttfamily\bfseries,
      breakatwhitespace=false,         
      breaklines=false,                 
      captionpos=b,                    
      keepspaces=true,     
      numbers=none,                    
      numbersep=4pt,                  
      showspaces=false,                
      showstringspaces=false,
      showtabs=false,                  
      tabsize=2,
      language=Java,
      escapechar=\%,
      moredelim=**[is][\btHL]{`}{`},
	  moredelim=**[is][{\btHL[fill=red!40]}]{@}{@},
   }
\definecolor{light-gray}{gray}{0.80}

\xpatchcmd{\proof}{.}{\proofpunctuation}{}{}
\xpatchcmd{\proof}{\itshape}{\prooffont}{}{}
\xpretocmd{\proof}{\setlength{\parindent}{0pt}}{}{}

\newcommand{\proofpunctuation}{.}
\newcommand{\prooffont}{\bfseries}

\begin{document}

\title{Learning Semantic Program Embeddings with Graph Interval Neural Network}

\author{Yu Wang}
\affiliation{
  \position{Position1}
  \department{Department of Computer Science and Technology}              
  \institution{Nanjing University}            
  \city{Nanjing}
  \state{Jiangsu}
  \postcode{210023}
  \country{China}                    
}
\email{yuwang_cs@smail.nju.edu.cn}          

\author{Fengjuan Gao}
\affiliation{
  \position{Position1}
  \department{Department of Computer Science and Technology}              
  \institution{Nanjing University}            
  \city{Nanjing}
  \state{Jiangsu}
  \postcode{210023}
  \country{China}                    
}
\email{fjgao@smail.nju.edu.cn}          

\author{Linzhang Wang}
\affiliation{
  \position{Position1}
  \department{Department of Computer Science and Technology}              
  \institution{Nanjing University}            
  \city{Nanjing}
  \state{Jiangsu}
  \postcode{210023}
  \country{China}                    
}
\email{lzwang@nju.edu.cn}          

\author{Ke Wang}
\affiliation{
  \institution{Visa Research.}            
  \city{Palo Alto}
  \state{CA}
  \country{USA}    
}
\email{kewang@visa.com}          

\begin{abstract}
Learning distributed representations of source code has been a challenging task for machine learning models. Earlier works treated programs as text so that natural language methods can be readily applied. Unfortunately, such approaches do not capitalize on the rich structural information possessed by source code. Of late, Graph Neural Network (GNN) was proposed to learn embeddings of programs from their graph representations. Due to the homogeneous (\ie do not take advantage of the program-specific graph characteristics) and expensive (\ie require heavy information exchange among nodes in the graph) message-passing procedure, GNN can suffer from precision issues, especially when dealing with programs rendered into large graphs. In this paper, we present a new graph neural architecture, called Graph Interval Neural Network (GINN), to tackle the weaknesses of the existing GNN. Unlike the standard GNN, GINN generalizes from a curated graph representation obtained through an abstraction method designed to aid models to learn. In particular, GINN focuses exclusively on intervals (generally manifested in looping construct) for mining the feature representation of a program, furthermore, GINN operates on a hierarchy of intervals for scaling the learning to large graphs.

We evaluate GINN for two popular downstream applications: variable misuse prediction and method name prediction. Results show in both cases GINN outperforms the state-of-the-art models by a comfortable margin. We have also created a neural bug detector based on GINN to catch null pointer deference bugs in Java code. While learning from the same 9,000 methods extracted from 64 projects, GINN-based bug detector significantly outperforms GNN-based bug detector on 13 unseen test projects.
Next, we deploy our trained GINN-based bug detector and Facebook Infer, arguably the state-of-the-art static analysis tool, to scan the codebase of 20 highly starred projects on GitHub. Through our manual inspection, we confirm 38 bugs out of 102 warnings raised by GINN-based bug detector compared to 34 bugs out of 129 warnings for Facebook Infer.
We have reported 38 bugs GINN caught to developers, among which 11 have been fixed and 12 have been confirmed (fix pending). GINN has shown to be a general, powerful deep neural network for learning precise, semantic program embeddings.
\end{abstract} 
\maketitle

\section{Introduction}
\label{sec:intro}

Learning distributed representations of source code has attracted growing interest and attention in program language research over the past several years. Inspired by the seminal work word2vec~\cite{Mikolov:2013}, many methods have been proposed to produce vectorial representation of programs. Such vectors, commonly known as program embeddings, capture the semantics of a program through their numerical components such that programs denoting similar semantics will be located in close proximity to one another in the vector space. A benefit of learning program embedding is to enable the application of neural technology to a board range of programming language (PL) tasks, which were exclusively approached by logic-based, symbolic methods before. 

\vspace*{3pt}
\noindent
\textbf{\textit{Existing Models of Source Code.}}\,
Earlier works predominately considered source code as text, and mechanically transferred natural language methods to discover their shallow, textual patterns~\cite{Hindle10,Pu2016,AAAI1714603}. Following approaches aim to learn semantic program embeddings from Abstract Syntax Tree (AST)~\cite{maddison2014structured,Alon:2019:CLD:3302515.3290353,alon2018code2seq}. Despite the significant progress, tree-based models are still confined to learning syntactic program features due to the limited offering from AST. Recently, another neural architecture, called Graph Neural Network (GNN), has been introduced to the programming domain~\cite{li2015gated,allamanis2017learning}. The idea is to reduce the problem of learning embeddings from source code to learning from their graph representations.
Powered by a general, effective message-passing procedure~\cite{Gilmer10},
GNN has achieved state-of-the-art results in a variety of problem domains: variable misuse prediction~\cite{allamanis2017learning}, program summarization~\cite{fernandes2018structured}, and bug detection and fixing~\cite{Dinella2020HOPPITY}. Even though GNN has substantially improved the prior works, their precision and efficiency issues remain to be dealt with. First, GNN is indiscriminate in which types of graph data they deal with. In other words, once programs are converted into graphs, they will be processed in the same manner as any other graph data (\eg social networks, molecular structures). As a result, it misses out on the opportunity to capitalize on the unique graph characteristics possessed by programs. Second, perhaps more severely, GNN has difficulties in learning from large graphs due to the high cost of its underlying message-passing procedure. Ideally, every node should pass messages directly or indirectly to every other node in the graph to allow sufficient information exchange. However, such an expensive propagation is hard to scale to large graphs without incurring a significant precision loss.

\vspace*{3pt}
\noindent
\textbf{\textit{Graph Interval Neural Network.}}\,
In this paper, we present a novel, general neural architecture called Graph Interval Neural Network (GINN) for learning semantic embeddings of source code. The design of GINN is based on a key insight that by learning from abstractions of programs, models can focus on code constructs of greater importance, and ultimately capture the essence of program semantics in a precise manner. At the technical level, we adopt GNN's learning framework 
for their cutting-edge performance. An important challenge we need to overcome is how to abstract programs into more efficient graph representations for GINN to learn. To this end,
we develop a principled abstraction method directly applied on control flow graphs. In particular, we derive from the control flow graph a hierarchy of intervals --- subgraphs that generally represent looping constructs --- as a new form of program graphs. Under this abstracted graph representation, GINN not only focuses exclusively on intervals for extracting deep, semantic program features but also scales its generalization to large graphs efficiently. Note that GINN's scalability-enhancing approach is fundamentally different from the existing techniques~\cite{allamanis2017learning}, where extra edges are required to link nodes that are far apart in a graph. Moreover,~\citet{allamanis2017learning} have not quantified the improvement of model scalability due to the added edges, let alone identify which edges to add \wrt the different graphs and downstream tasks.



We have realized GINN as a new model of source code and extensively evaluated it. To demonstrate its generality, we evaluate GINN in multiple downstream tasks. First, we pick two prediction problems defined by the prior works~\cite{allamanis2017learning,Alon:2019:CLD:3302515.3290353} in which  
GINN significantly outperforms RNN Sandwich and sequence GNN, the state-of-the-art models
in predicting variable misuses and method names respectively.
Second, we evaluate GINN in detecting null pointer dereference bugs in Java code, a task that is barely attempted by learning-based approaches.
Results show GINN-based bug detector is significantly more effective than the baseline built upon GNN. We also deploy our trained bug detector and Facebook Infer~\cite{calcagno2015moving,berdine2005smallfoot}, arguably the state-of-the-art static analysis tool, to search bugs in 20 Java projects that are among the most starred on GitHub. After manually inspecting the alarms raised by both tools, we find that our bug detector yields a higher precision (38 bugs out of 102 alarms) than Infer (34 bugs out of 129 alarms). We have reported the 38 bugs detected by GINN to developers, out of which 11 are fixed and 12 are confirmed (fix pending). Our evaluation suggests
GINN is a general, powerful deep neural network for learning precise, semantic program embeddings. 

\vspace*{3pt}
\noindent
\textbf{\textit{Contributions.}}\, Our main contributions are:
\vspace*{-2pt}
\begin{itemize}
    \item We propose a novel, general deep neural architecture, Graph Interval Neural Network,
    which generalizes from a curated graph representation, obtained through a graph abstraction method on the control-flow graph.
    
    \item We realize GINN as a new deep model of code, which readily serves multiple downstream PL tasks. More importantly, GINN outperforms the state-of-the-art model by a comfortable margin in each downstream task.
    
    \item We present the details of our extensive evaluation of GINN in variable misuse prediction, method name prediction, and null pointer dereference detection. 
    
    \item We publish our code and data at \url{https://figshare.com/articles/datasets_tar_gz/8796677} to aid future research activity.
\end{itemize}

\section{Preliminary}
\label{sec:pre}

In this section, we briefly revisit connected, directed graphs, interval~\cite{Allen1970} and GNN~\cite{gori2005new}, which lay the foundation for our work.

\subsection{Graph}
\label{subsection:gra}
A graph $G\! =\! (V, E)$ consists of a set of nodes $V \!= \!\{v_1, ..., v_n\}$, and a list of directed edge sets $E = (E_1, . . . , E_K)$ where $K$ is the total number of edge types and $E_k$ ($1 \leq k \leq K$) is a set of edges of type $k$. We denote by $(v_s, v_d, k) \in E_k$ an edge of type $k$ directed from node $v_s$ to node $v_d$. For graphs with only one edge type, an edge is represented as $(v_s, v_d)$.

The immediate successors of a node $v_i$, denoted by $\textit{post}\,(v_i)$, are all of the nodes $v_j$ for which $(v_i, v_j)$ is an edge in one edge set in $E$. The immediate predecessors of node $v_j$, denoted by $\textit{pre}\,(v_j)$, are all of the nodes $v_i$ for which $(v_i, v_j)$ is an edge in one edge set in $E$. 

A path is an ordered sequence of nodes $(v_p,..., v_q)$ and their connecting edges, in which each node $v_t$, $t\in(p,\dots,q-1)$, is an immediate predecessor of $v_{t+1}$. A closed path is a path in which the first and last nodes are the same. The successors of a node $v_t$, denoted by $\textit{post}\,^{*}(v_t)$, are all of the nodes $v_x$ for which there exists a path from $v_t$ to $v_x$. The predecessors of a node $v_t$, denoted by $\textit{pre}\,^{*}(v_t)$, are all of the nodes $v_y$ for which there exists a path from $v_y$ to $v_t$.

\subsection{Interval}
Introduced by~\citet{Allen1970}, an interval \textit{I(h)} is the maximal,
single entry subgraph in which \textit{h} is the only entry node
and all closed paths contain \textit{h}. The unique
interval node \textit{h} is called the interval head or simply the
header node. An interval can be expressed in terms of
the nodes in it: $I(h) = \{v_l, v_2, ... ,v_m\}.$ 

By selecting the proper set of header nodes, a graph can be partitioned into a set of disjoint intervals. We show the partition procedure proposed by~\citet{Allen1970} in Algorithm \ref{alg:intervals}. The key is to add to an interval a node only if all of its immediate predecessors are already in the interval (Line \ref{li:w2} to \ref{li:s2}). The intuition is such nodes when added to an interval keep the original header node as the single entry of an interval. To find a header node to form another interval, a node is picked that is not a member of any existing intervals although it must have a (but not all) immediate predecessor being a member of the interval that is just computed
(Line \ref{li:w3} to \ref{li:s3}). We repeat the process until reaching the fixed-point where all nodes are members of an interval. 

The intervals on the original graph 
are called the first order intervals, denoted by 
$I^1(h)$, and the graph from which they were derived 
is called first order graph $G^1$. Partitioning the first 
order graph results in a set of first order intervals, denoted by $\mathcal{S}^1$ \st 
$I^1(h) \in \mathcal{S}^1$. By making each first order 
interval into a node and each interval exit edge into 
an edge, the second order graph can be derived, from 
which the second order intervals can also be defined. 
The procedure can be repeated to derive successively 
higher order graphs until the n-th order graph consists 
of a single interval.
Figure~\ref{fig:interval} illustrates 
such a sequence of derived graphs. 

\begin{minipage}{.99\linewidth}
\hspace{-12.5pt}
\begin{algorithm}[H]
\small
\caption{Finding intervals for a given graph} 
\label{alg:intervals}
\raggedright
\KwIn {a set of nodes $V$ on a graph}
\KwOut {a set of intervals $\mathcal{S}$}
// Assume $v_0$ is the unique entry node of the graph

H = \{$v_0$\}\;
\While{$H \neq \emptyset$}
{
        // remove next $h$ from $H$
        
        h = H.pop()\;
        I(h) = \{h\}\;
        // only nodes that are neither in the current interval nor any other interval will be considered

        \While{$\{v \in V \,| \, v \notin I(h) \wedge \nexists s (s \in \mathcal{S} \wedge v \in s) \wedge  pre(v) \subseteq I(h)\} \neq \emptyset$\label{li:w2}}   
        {
           I(h) = I(h) $\cup$ \{ v \}\; \label{li:s2}
        }
        // find next headers
        
        \While{$\{v \in V  \,| \, \nexists s_{1} (s_{1} \in \mathcal{S} \wedge v \in s_{1})  \wedge \exists m_{1}, m_{2} (m_{1},m_{2} \in pre(v) \wedge m_{1} \in I(h) \wedge m_{2} \notin I(h)) \}\! \neq\! \emptyset$\label{li:w3}}
        {
            H = H $\cup$ \{ v \}\; \label{li:s3}
        }
        $\mathcal{S}$ = $\mathcal{S}$ $\cup$ I(h)\;
}
\end{algorithm}
\end{minipage}
\begin{figure}[h]
\vspace*{-17pt}
\centering
  \hspace{-4pt}\includegraphics[width=.992\linewidth]{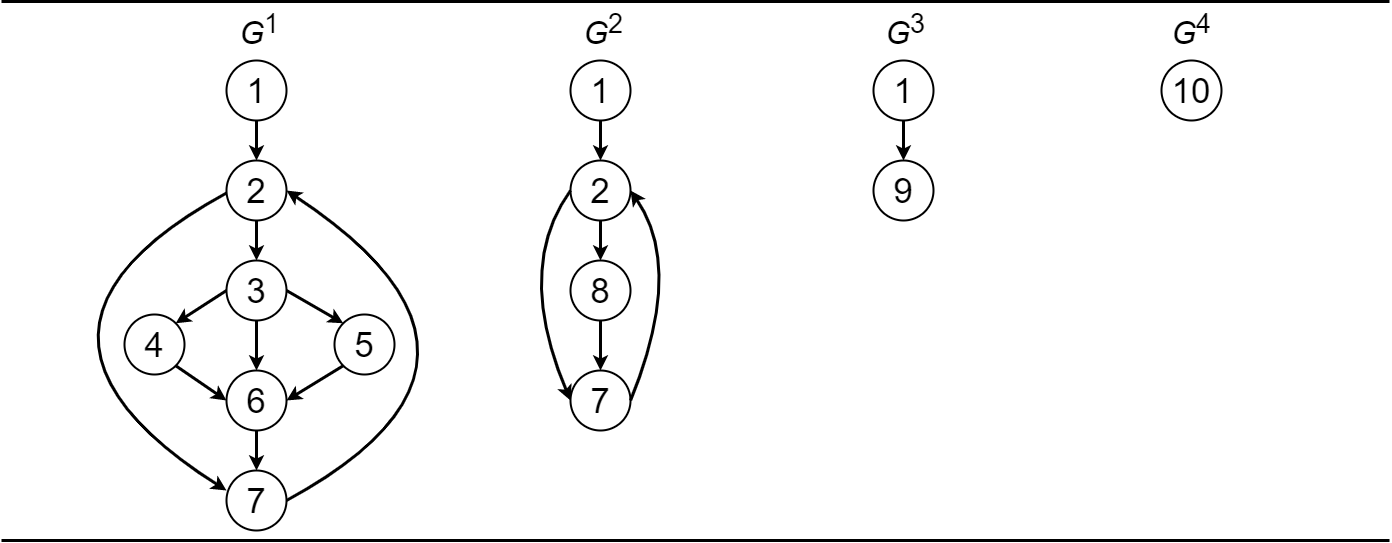}
  \caption{n-th order intervals and graphs. The set of intervals on $\mathcal{S}^1$ are $I^1(1)$=\{1\}, $I^1(2)$=\{2\}, $I^1(3)$=\{3,4,5,6\}, $I^1(7)$=\{7\}. 
  The set of intervals on $\mathcal{S}^2$ are $I^2(1)$=\{1\} and $I^2(2)$=\{2,7,8\}. $I^3(1)$=\{1,9\} and $I^4(10)$=\{10\} are the only intervals on $\mathcal{S}^3$ and $\mathcal{S}^4$ respectively.}
  \label{fig:interval}
\end{figure}
\vspace*{-20pt}

\subsection{Graph Neural Network}
\label{subsec:gnn}
Graph Neural Network (GNN)~\cite{gori2005new} is a specialized machine learning model designed to learn from graph data. The intuitive idea underlying GNN is that nodes in a graph represent objects and edges represent their relationships. Thus, each node $v$ can be attached to a vector, called state, which collects a representation of the object denoted by $v$. Naturally, the state of $v$ can be specified using the information provided by nodes in the neighborhood of $v$. Technically, there are many realizations of this idea. Here, we describe the message-passing technique~\cite{Gilmer10}, a widely applied method for GNN to compute the states of nodes.

We reuse the terminology of graph introduced in Section~\ref{subsection:gra}. Given a graph ${G}\! =\! ({V}, {E})$ where ${V}$ and ${E}$ are the inputs, a neural message-passing GNN computes the state vector for each node in a sequence of steps. In each step, every node first sends messages to all of its neighbors, and then update its state vector by aggregating the information received from its neighbors.
\begin{equation}
\mu^{(l+1)}_{v} = \phi(\{\mu^{(l)}_{u}\}_{u \in \mathcal{N}(v)})  \label{equ:h} %
\end{equation}
$\mu^{(l+1)}_{v}$ denotes the state of $v$ in $l+1$ steps, which is determined by the state of its neighbors in the previous step. $\mathcal{N}{(v)}$ denotes the neighbors that are connected to $v$. Formally, $\mathcal{N}{(v)} = \{u|(u,v,k) \in {E}_k, \forall k \in \{1,2,...,K\}\}$. $\phi(\cdot)$ is a non-linear function that performs the aggregation. After the state of each node is computed with many rounds of message passing, a graph representation can also be obtained through various pooling operations (\eg average or sum). 

Some GNN~\cite{Si10} computes a separate node state \wrt an edge type (\ie $\mu^{(l+1), k}_{v}$ in Equation~\ref{equ:type}) before aggregating them into a final state (\ie $\mu^{(l+1)}_{v}$ in Equation~\ref{equ:sum}).
\begin{align*}
\mu^{(l+1), k}_{v} &= \phi_1(\sum_{\mathclap{u \in \mathcal{N}^{k}{(v)}}} \mathbf{W_1} \mu_u^{(l)}), \forall k \in \{1,2,...,K\} \Label{equ:type} &\;\;\;\;\;\;\,
\mu^{(l+1)}_{v} &= \phi_2(\mathbf{W_2}[ \mu_v^{(l), 1}, \mu_v^{(l), 2}, ..., \mu_v^{(l), K}]) \Label{equ:sum}
\end{align*}
$\mathbf{W_1}$ and $\mathbf{W_2}$ are variables to be learned, and $\phi_1$ and $\phi_2$ are some nonlinear activation functions.

\citet{li2015gated} proposed Gated Graph Neural Network (GGNN) as a new variant of GNN. Their major contribution is a new instantiation of $\phi(\cdot)$ using Gated Recurrent Units~\cite{cho-cho2014learning}. The following equations describe how GGNN works:
\begin{align*}
\qquad\qquad\,\,
{m}_v^l &= \sum_{\mathclap{u \in \mathcal{N}{(v)}}} f(\mu_{u}^{(l)}) \qquad\qquad\Label{equ:mes1} & 
\qquad\qquad\qquad
\mu^{(l+1)}_{v} &= \mathit{GRU}({m}_v^l, \mu^{(l)}_{v})\qquad\qquad\Label{equ:mes2}
\end{align*}

To update the state of node $v$, Equation~\ref{equ:mes1} computes a message ${m}_v^l$ using $f(\cdot)$ (\eg a linear function) from the states of its neighboring nodes $\mathcal{N}{(v)}$. Next, a $GRU$ takes ${m}_v^l$ and $\mu^{(l)}_{v}$ --- the current state of node $v$ --- to compute the new state $\mu^{(l+1)}_{v}$ (Equation~\ref{equ:mes2}). 

\vspace*{3pt}
\noindent
\textbf{\textit{GNN's Weakness.}}\,
Although the message-passing technique has highly empowered GNN, it severely hinders GNN's capability of learning patterns formed by nodes that are far apart in a graph. From a communication standpoint, 
a message 
sent out by the start node needs to go through many intermediate nodes en route to the end node. Due to the aggregation operation (regardless of which implementation in Equation~\ref{equ:h},~\ref{equ:sum}, or~\ref{equ:mes2} is adopted), the message gets diluted whenever it is absorbed by an intermediate node for updating its own state. By the time the message reaches the end node, it is too imprecise to bear any reasonable resemblance to the start node. This problem is particularly challenging in the programming domain where long-distance dependencies are common, important properties models have to capture.~\citet{allamanis2017learning} designed additional edges to connect nodes that exhibit a relationship of interest (\eg data or control dependency) but are otherwise far away from each other in an AST. The drawbacks are: first, their idea, albeit conceptually appealing, has not been rigorously tested in practice, therefore, the extent to which it addresses the scalability issue is unknown; second, they have not offered a principled solution regarding which edges to add in order to achieve the best scalability. Instead, they universally added a predefined set of edges to any graph, an approach that is far from satisfactory.
\section{Methodology and Approach}
\label{sec:met}

\begin{figure*}[t]
	\begin{minipage}{0.21\textwidth}
	    \includegraphics[width=\textwidth]{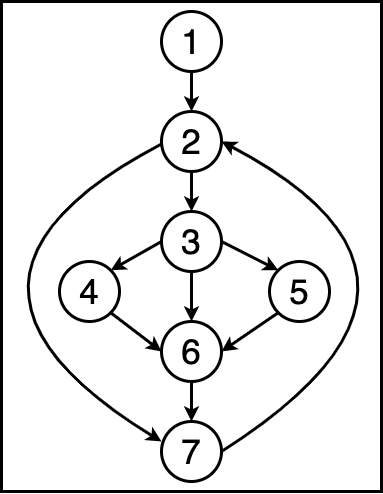}
        \caption{}
		\label{fig:1cfg}
	\end{minipage}
	\begin{minipage}{0.21\textwidth}	
	    \includegraphics[width=\textwidth]{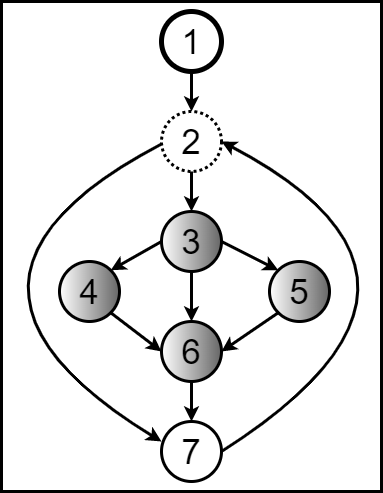}
        \caption{}
		\label{fig:1cfgs}
	\end{minipage}
    \begin{minipage}{0.1187\textwidth}
	    \includegraphics[width=\textwidth]{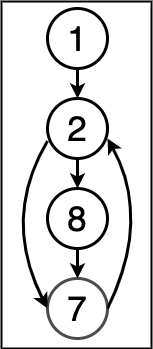}
        \caption{}
		\label{fig:2cfg}
	\end{minipage}
	\begin{minipage}{0.1187\textwidth}	
	    \includegraphics[width=\textwidth]{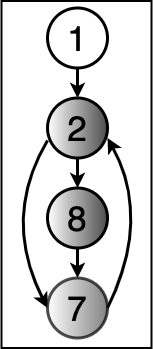}
        \caption{}
		\label{fig:2cfgs}
	\end{minipage}
	\begin{minipage}{0.1187\textwidth}	
	    \includegraphics[width=\textwidth]{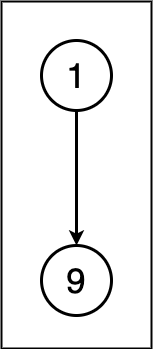}
        \caption{}
		\label{fig:3cfg}
	\end{minipage}	
\end{figure*}


Our intuition is that generalization can be made easier
when models are given abstractions (of programs) to learn.
Considering its generality and capacity, we choose GNN's learning framework to incorporate our abstraction methodology.
A key challenge arises: 
how to abstract programs into favorable graph representations for GNN to learn.  
A natural idea would be abstracting programs at the level of source code and then reducing the abstracted programs into graphs. On the other hand, a simpler and perhaps more elegant approach is to abstract programs directly on their graph representation, which this paper explores. Specifically, we propose a new graph model, called Graph Interval Neural Network (GINN), which uses control flow graph to represent an input program,\footnote{We discuss a few variations in Section~\ref{sec:eva}.} and abstracts it with three primitive operators: partitioning, heightening, and lowering. Each primitive operator is associated with an extra operator for computing the state of each node on the resultant graph after the primitive operator is applied. Using the control flow graph depicted in Figure~\ref{fig:1cfg} as our running example, we explain in details of GINN's learning approach. 


To start with, GINN uses the partitioning operator to split the graph into a set of intervals. Figure~\ref{fig:1cfgs} adopts different styling of the nodes to indicate the four intervals on the graph. Then, it proceeds to compute the node states using the extra operator. Unlike the standard message-passing mechanism underpinning the existing GNN, GINN restricts messages to be passed within an interval. If an interval is a node itself, its state will not get updated by the extra operator.
For example, node 3, 4, 5, and 6 in Figure~\ref{fig:1cfgs} will freely pass messages to the neighbors while evolving their own states. On the other hand, node 1 and 2, which are alone in their respective intervals, keep their states as they were.
The propagation that occurs within an interval is carried out in the same manner as it is in the existing GNN (Equation~\ref{equ:h},~\ref{equ:sum}, or~\ref{equ:mes2} depending on the actual implementation). We formalize the partitioning operator as follows:

\begin{definition}(\textit{Partitioning})
    Partitioning operator, denoted by $\rho$, is a sequence ($\rho_{\mathit{abs}}$, $\rho_{\mathit{com}}$) where $\rho_{\mathit{abs}} : {G}^n \rightarrow  \mathcal{S}^n$ is a function that maps an n-th order graph into a set of n-th order intervals, and
    $\rho_{\mathit{com}} : \forall I^{n}(h) \in \mathcal{S}^{n} \rightarrow \mu_{v_{1}}, \mu_{v_{2}},\dots \mu_{v_{T}} \! \in \mathbb{R}^{e} (v_{1},v_{2},\dots,v_{T} \in I^{n}(h))$ is a function that computes a $e$-dimensional vector for each node in each n-th order interval. Equation~\ref{equ:part} defines $\rho_{\mathit{com}}$.
\end{definition}
Depends on how many nodes $I^{n}(h)$ is composed of, two cases to be considered:
\begin{equation}\label{equ:part}
    \begin{cases} 
    \mu^{(l+1)}_{v} = \phi(\{\mu^{(l)}_{u}\}_{u \in \mathcal{N}(v)}) & v \in I^{n}(h) \wedge u \in I^{n}(h) \\
    \mu^{(l+1)}_{v} = \mu^{(l)}_{v} & v \in I^{n}(h) \wedge \lvert I^{n}(h) \rvert = 1   
    \end{cases}
\end{equation}

Next, GINN applies the heightening operator to convert the first order graph into the second order graph (Figure~\ref{fig:2cfg}). In particular, the heightening operation replaces each active interval (of multiple nodes) on the lower order graph with a single node on the higher order graph (\eg node 3, 4, 5, and 6 are replaced with node 8). In terms of the states of the freshly created nodes, GINN initializes them to be the weighted sum of the states of replaced nodes. To motivate this design decision, we also consider simply averaging the states of replaced nodes in Section~\ref{subsubsec:res} and~\ref{subsubsec:metres}.
The heightening operator is defined below.

\begin{definition}(\textit{Heightening})
    Heightening operator, denoted by $\eta$, is a sequence ($\eta_{\mathit{abs}}$, $\eta_{\mathit{com}}$) where $\eta_{\mathit{abs}} : {G}^n \rightarrow  {G}^{n+1}$ is a function that maps an n-th order graph into an (n+1)-th order graph, and $\eta_{\mathit{com}} : \mu_{v^{n}_{1}}, \mu_{v^{n}_{2}},\dots \mu_{v^{n}_{N}} \in \mathbb{R}^{e} \rightarrow \mu_{v^{n+1}_{1}}, \mu_{v^{n+1}_{2}},\dots \mu_{v^{n+1}_{N^\prime}} \! \in \mathbb{R}^{e} (\{v^{n}_{1},v^{n}_{2},\dots v^{n}_{N}\}$ denotes the set of nodes on $G^{n}$; $\{v^{n+1}_{1},v^{n+1}_{2},\dots v^{n+1}_{N^\prime}\}$ denotes the set of nodes on $G^{n+1})$ is a function that maps a set of $e$-dimensional vectors for each node on ${G}^n$ to another set of $e$-dimensional vectors for each node on ${G}^{n+1}$. Equation~\ref{equ:merge} defines $\eta_{\mathit{com}}$. 
\end{definition}
Like Equation~\ref{equ:part}, we consider two cases: 
\begin{align*}
\quad\,\,\,
    \begin{cases} 
    \mu^{(0)}_{v_{i}^{n+1}} = \mu^{(l)}_{v_{i}^{n}} & \lvert I^{n}(h) \rvert = 1 \wedge v_{i}^{n} \in I^{n}(h) \\
      \mu^{(0)}_{v_{h}^{n+1}} = \sum_{v \in I^{n}(h)} \alpha_{v}\mu^{(l)}_{v} & \lvert I^{n}(h) \rvert > 1   
    \end{cases}\quad\! \Label{equ:merge}& \qquad\qquad
    \alpha_{v} = \frac{exp\,(\lvert\lvert\mu^{(l)}_{v}\rvert\rvert)}{\sum\limits_{\mathclap{v^\prime \in I^{n}(h)}} exp\,(\lvert\lvert\mu^{(l)}_{v^\prime}\rvert\rvert)} \Label{equ:weight}
\end{align*}
given that $v_{i}^{n+1}$ on the (n+1)-th order graph is the same node as $v_{i}^{n}$ on the n-th order graph, and $v_{h}^{n+1}$ is the node that replaces the interval $I^{n}(h)$. The weight of each state vector, $\alpha_{v}$, is determined based on its length
(Equation~\ref{equ:weight}). 

Now GINN uses the partitioning operator again to produce a set of intervals from the second order graph (Figure~\ref{fig:2cfgs}). Due to the restriction outlined earlier, messages are only passed among node 2, 7, and 8. However, since node 8 can be deemed as a proxy of node 3\textasciitilde6, this exchange, in essence, covers all but node 1 on the first order graph. In other words, even though messages are always exchanged within an interval, heightening abstraction expands the region an interval covers, thus enabling the communication of nodes that were otherwise far apart. 

Then, GINN applies the heightening operator the second time, as a result, the third order graph emerges (Figure~\ref{fig:3cfg}).
Since the third order graph is an interval itself, 
GINN computes the state of node 1 and 9 in the same way as existing GNN using partitioning operator. Now, GINN reaches an important point, which we call \textit{sufficient propagation}. That is every node has passed messages either directly or indirectly to every other reachable node in the original control flow graph. In other words, the heightening operator is no longer needed. 

Upon completion of the message-exchange between node 1 and 9, GINN applies for the first time the lowering operator to convert the third order graph back to the second order graph. In particular, nodes that were created on the higher order graphs for replacing an interval on the lower order graph will be split back into the nodes of which the interval previously consists (\eg node 9 to node 2, 7, and 8). The idea is, after reaching sufficient propagation, GINN begins to recover the state of each node on the original control flow graph. We define the lowering operator below.

\begin{definition}(\textit{Lowering})
    Lowering operator, denoted by $\lambda$, is a sequence ($\lambda_{\mathit{abs}}$, $\lambda_{\mathit{com}}$) where $\lambda_{\mathit{abs}} : {G}^{n+1} \rightarrow  {G}^{n}$ is the inverse function of $\eta_{\mathit{abs}}$, which maps an (n+1)-th order graph into an n-th order graph, and $\lambda_{\mathit{com}} : \mu_{v^{n+1}_{1}}, \mu_{v^{n+1}_{2}},\dots \mu_{v^{n+1}_{N^\prime}} \in
    \mathbb{R}^{e} \rightarrow \mu_{v^{n}_{1}}, \mu_{v^{n}_{2}},\dots \mu_{v^{n}_{N}} \! \in \mathbb{R}^{e} (\{v^{n+1}_{1},v^{n+1}_{2},\dots v^{n+1}_{N^\prime}\}$ denotes the set of nodes on $G^{n+1}$; $\{v^{n}_{1},v^{n}_{2},\dots v^{n}_{N}\}$ denotes the set of nodes on $G^{n})$ is a function that maps a set of $e$-dimensional vectors for each node on ${G}^{n+1}$ to another set of $e$-dimensional vectors for each node on ${G}^{n}$. Equation~\ref{equ:split} defines $\lambda_{\mathit{com}}$.
\end{definition}
In general, $\lambda_{\mathit{com}}$ can be considered as a reversing operation of $\eta_{\mathit{com}}$:

\begin{equation}\label{equ:split}
    \begin{cases} 
    \mu^{(0)}_{v_{i}^{n}} = \mu^{(l)}_{v_{i}^{n+1}} & \lvert I^{n+1}(h) \rvert = 1 \wedge v_{i}^{n+1}\! \in\! I^{n+1}(h) \\
    \mu^{(0)}_{v} = \alpha_{v}\mu^{(l)}_{v_{h}^{n+1}} * \lvert I^{n+1}(h) \rvert & \lvert I^{n+1}(h) \rvert > 1   
    \end{cases}
\end{equation}
given that $v_{i}^{n}$ on the n-th order graph is the same node as $v_{i}^{n+1}$ on the (n+1)-th
order graph, and $v_{h}^{n+1}$ is the node that is split into a multitude of nodes, each of which is assigned $\alpha_{v}$, the weight defined in Equation~\ref{equ:weight}. $\lvert I^{n+1}(h) \rvert$ is the regulation term that aims to bring node $v$'s initial state and its final state (prior to being replaced by the heightening operator) within the same order of magnitude.


As GINN revisits the second order graph, it re-applies the partitioning operator to produce the only interval to be processed (\ie node 2, 7, and 8). However, there lies a crucial distinction in computing the states of node 2, 7, and 8 this time around. That is the communication among the three nodes is no longer local to their own states, but also inclusive of node 1's thanks to the exchange between node 1 and 9 on the third order graph. In general, such recurring interval propagation improves the precision of node representations by incorporating knowledge gained about the entire graph.

Next, GINN applies the lowering operator again to recover the original control flow graph. After the partitioning operation, GINN computes the state vector for every node on the graph, signaling the completion of a whole learning cycle. The process will be repeated for a fixed number of steps. Then, we use the state vectors from the last step as the node representations. We generalize the above explanation with the running example to formalize GINN's abstraction cycle.

\begin{definition}(\textit{Abstraction Cycle})
	GINN's abstraction cycle is a sequence composed of three 
	primitive operators: partitioning, denoted by $\rho$; heightening, 
	denoted by $\eta$; and lowering, denoted by $\lambda$. The sequence 
	can be written in the form of 
	$((\rho \rightarrow \eta)^{*} \rightarrow \rho^{\prime} \rightarrow (\lambda \rightarrow \rho)^{*})^{*}$ where 
	$\lvert(\rho \rightarrow \eta)^{*}\rvert=\lvert(\lambda \rightarrow \rho)^{*}\rvert$.  
    $\rho^{\prime}$ denotes the partitioning operator on the highest order graph in which 
    case the graph is a single interval itself.
\end{definition}

To sum up, GINN adopts control flow graphs as the program representation. At first, it uses heightening operator to increase the order of graphs, as such, involving more nodes to communicate within an interval. Later, by applying the lowering operator, GINN restores local propagation on the reduced order of graphs, and eventually recovers the node states on the original control flow graph. Worth mentioning heightening and lowering operators are complementary and their cooperation benefits the overall GINN's learning objective. Specifically, heightening operator enables GINN to capture the global properties of a graph, with which lowering operator helps GINN to compute precise node representations.
As the interleaving between heightening and lowering operator goes on, node representations will continue to be refined to reflect the properties of a graph.

\vspace*{3pt}
\noindent
\textbf{\textit{Abstracting the Source Code.}}\, We aim to demystify what the above graph abstraction method translates to at the source code level, and why it helps GINN to learn. Since intervals play a pivotal role in all three operators of the graph abstraction method,
it is imperative that we understand what an interval represents in a program. As a simple yet effective measure, we use a program, which has the same control flow graph as the one in Figure~\ref{fig:1cfg}, to disclose what code construct an interval correlates to, and more importantly, how the abstraction method works at the source code level. The program is depicted in Figure~\ref{fig:tran}, which computes the starting indices of the substring \texttt{s1} in string \texttt{s2}. 
Recall the only interval GINN processed on the first order graph, the one that consists of node 3\textasciitilde6, which represents the inner \texttt{do-while} loop in the function. In other words, GINN focuses exclusively on the inner loop when learning from the first order graph. Similarly, node 2, node 7, and node 8 consist of the only interval from which GINN learns at the second order graph. This interval happens to be the nested \texttt{do-while} loop. Finally, GINN spans across the whole program as node 1 and node 9 cover every node in the first order graph. Figure~\ref{fig:tran} illustrates GINN's transition among the first, second, and third order graph, in particular, the distinct focus on each order graph. Inspired from the exposition above, we present Theorem~\ref{the:ci} and Corollary~\ref{cor:hi}.

\lstset{style=mystyle}
\begin{figure*}
        \centering
        \begin{tikzpicture}
        \matrix[matrix of nodes, 
        nodes={anchor=center},column sep=-1.5em]{%
            \begin{minipage}[t]{0.28\textwidth}
\begin{lstlisting}[linewidth=3.84cm,xleftmargin=-3em,framexleftmargin=-0.3em]
void substringIndices(%\color{blue}{String}% s1, 
	%\color{blue}{String}% s2, int attempts, int[] arr) 
{
	%\color{light-gray}{int M = s1.length();}%
	%\color{light-gray}{int N = s2.length();}%
	%\color{light-gray}{int i = 0,j = 0,res = 0,pos = 0;}%			

	%\color{light-gray}{do \{}%
		%\color{light-gray}{if (pos < attempts) \{}%
			do {
				log("Begin at index: " + i);
				if (s2.charAt(i+j) != 
						s1.charAt(j)) {
					res =  -1;
				  log("Differs at: " + j);}
				else
					log("Equal at: " + j);
				j++;
			} while(j < M);
    %\color{light-gray}{\}}%
		
		%\color{light-gray}{arr[pos++] = res;}%
		%\color{light-gray}{i++; j = 0;}%
		%\color{light-gray}{res = pos < attempts ? i : 0;}%		
	%\color{light-gray}{\} while (i <= N - M);}%
%\color{light-gray}{\}}%
\end{lstlisting}
            \end{minipage}
            &
            \hspace{1em}
            \includegraphics[width=0.03\textwidth]{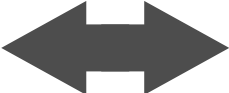}
            &
            \hspace{.8em}
            \begin{minipage}[t]{0.28\textwidth}
\begin{lstlisting}[linewidth=4.5cm,framexleftmargin=-0.3em]
void substringIndices(%\color{blue}{String}% s1, 
	%\color{blue}{String}% s2, int attempts, int[] arr) 
{
	%\color{light-gray}{int M = s1.length();}%
	%\color{light-gray}{int N = s2.length();}%
	%\color{light-gray}{int i = 0,j = 0,res = 0,pos = 0;}%
	
	do {
		if (pos < attempts) {
			do {
				log("Begin at index: " + i);
				if (s2.charAt(i+j) != 
						s1.charAt(j)) {
					res =  -1;
				  log("Differs at: " + j);}
				else
					log("Equal at: " + j);
				j++;
			} while(j < M);
    }
		
		arr[pos++] = res;
		i++; j = 0;
		res = pos < attempts ? i : 0;
	} while (i <= N - M);
%\color{light-gray}{\}}%
\end{lstlisting}
            \end{minipage}
            &
            \hspace{1em}
            \includegraphics[width=0.03\textwidth]{pics/arrowd.png}
            &
            \hspace{.8em}
            \begin{minipage}[t]{0.28\textwidth}
\begin{lstlisting}[linewidth=4.5cm,framexleftmargin=-0.3em]
void substringIndices(%\color{blue}{String}% s1, 
	%\color{blue}{String}% s2, int attempts, int[] arr) 
{
	int M = s1.length();
	int N = s2.length();
	int i = 0,j = 0,res = 0,pos = 0;
	
	do {
		if (pos < attempts) {
			do {
				log("Begin at index: " + i);
				if (s2.charAt(i+j) != 
						s1.charAt(j)) {
					res =  -1;
				  log("Differs at: " + j);}
				else
					log("Equal at: " + j);
				j++;
			} while(j < M);
    }
		
		arr[pos++] = res;
		i++; j = 0;
    res = pos < attempts ? i : 0;
	} while (i <= N - M);
}
\end{lstlisting}
            \end{minipage}
            \\
        };
        \end{tikzpicture}
        \setlength{\abovecaptionskip}{-8pt}
        \setlength{\belowcaptionskip}{-4pt}
        \caption{GINN's focus in source code when transitioning back and forth between different order graphs.\label{fig:tran}}
\end{figure*}

\begin{theorem}[Interval Constituents]
\label{the:ci}
Given a loop structure $l$ on a control flow graph, $l$'s loop header, denoted by $h$, is the entry node of an interval that does not contain any predecessor of $h$ outside of $l$. However, the interval contains every node inside of $l$ assuming $l$ does not contain an inner loop construct.
\end{theorem}

\begin{proof}
First, we prove the non-coexistence of $h$ and its predecessors outside of $l$. Assume otherwise, there is an interval $I$ containing at least one predecessor of $h$ outside of $l$ and $h$ itself. By construction, $I$ already had an entry node $h^\prime$. Due to the loop-forming back edge connecting node $m$ back to $h$ within $l$, $I$ now has two entry nodes, $h^\prime$ and $h$, if $m$ is not already in $I$, which violates the single-entry rule of an interval. For the other possibility that if $m$ is in $I$, then by construction all immediate predecessors of $m$ are also in $I$ given that $l$ does not contain an inner loop, recursively all nodes in $l$ are in $I$, meaning there will be closed paths within $I$ that do not go through $h^\prime$ (\ie those going through $h$), which violates the definition of interval.

Second, we prove $h$ is the entry node of the interval it is in. Assume otherwise, there exists an interval $I$, which $h$ is a part of, having a different entry node $h^\prime$. Since none of $h$'s predecessors outside of $l$ are part of $I$, at least one of them will connect to $h$, thus making $h$ the other entry node of $I$ in addition to $h^\prime$. Therefore, the single-entry rule of an interval is violated.

Finally, we prove that all nodes of $l$ are in the same interval as $h$ if $l$ does not contain an inner loop construct. Assume otherwise, there exists at least one node in $l$ that is not in the interval, $I$, for which $h$ is the entry node. By construction, at least one immediate predecessor of the node not in $I$ is also not in $I$. Since $l$ does not contain an inner loop construct, none of the nodes in $l$ has a predecessor outside of $l$. By recursion, ultimately, the loop header $h$ will be the only immediate predecessor of a node that is not in $I$. Therefore, $h$ is also not in $I$, which contradicts the assumption.
\end{proof}

\vspace*{-7pt}
\begin{corollary}[Hierarchical Intervals for Nested Loops]
\label{cor:hi}
Given a nested loop structure on a control flow graph, nodes of the n-th most inner loop lie in an interval on the n-th order graph.
\end{corollary}

\vspace*{-7pt}
\begin{proof}
Let P(n) be the statement that nodes of the n-th most inner loop lie in an interval on the n-th order graph. We give a proof by induction on n.

\textbf{{Base case}}: Show that the statement holds for the smallest number n = 1.\\
According to Theorem~\ref{the:ci}, P(0) is clearly true: nodes of the most inner loop lie in an interval on the first order graph.

\textbf{{Inductive step}}: Show that for any k $\geq$ 1, if P(k) holds, then P(k+1) also holds.\\
Assume the induction hypothesis that for a particular k, the single case n = k holds, meaning P(k) is true: nodes of the k-th most inner loop lie in an interval on the k-th order graph. By definition, nodes of the k-th most inner loop will be rendered into a single node within the (k+1)-th most inner loop on the (k+1)-th order graph. According to Theorem~\ref{the:ci}, we deduce that nodes of the (k+1)-th most inner loop will also lie in an interval on the (k+1)-th order graph. That is, the statement P(k+1) also holds, establishing the inductive step.

\textbf{{Conclusion}}: Since both the base case and the inductive step have been shown, by mathematical induction the statement P(n) holds for every number n.
\end{proof}

\begin{figure*}[tbp!]
    \adjustbox{max width=0.495\textwidth}{
	\begin{minipage}{0.5\textwidth}
    	\lstset{style=mystyle}
\begin{lstlisting}[linewidth=6.8cm,morekeywords={var, as}]
var clazz = classTypes["Root"].Single() as ...
Assert.NotNull(clazz);

var first = classTypes["RecClass"].Single() as ...
Assert.NotNull(@clazz@);

Assert.Equal("string", first.Properties["Name"].Name);
Assert.False(clazz.Properties["Name"].IsArray);
\end{lstlisting}
        \setlength{\abovecaptionskip}{-1pt}
        \setlength{\belowcaptionskip}{-8pt}
	    \caption{An example (extracted from ~\cite{allamanis2017learning}) of variable misuse. \texttt{clazz} highlighted in red is misused. Instead, \texttt{first} should have been used.}
	    \label{fig:misex}
	\end{minipage}}
    \adjustbox{max width=0.495\textwidth}{
	\begin{minipage}{0.5\textwidth}	
	    \lstset{style=mystyle}
\begin{lstlisting}[linewidth=6.8cm,morekeywords={var, public, string}]
public string[] f(string[] array)
{
    string[] newArray = new string[array.Length];
    for (int index = 0; index < array.Length; index++)
        newArray[array.Length-index-1] = array[index];

    return newArray;
}
\end{lstlisting}
        \setlength{\abovecaptionskip}{-1pt}
        \setlength{\belowcaptionskip}{-8pt}
	    \caption{An example function extracted from~\cite{Alon:2019:CLD:3302515.3290353} whose name is left out for models to predict. The answer in this case is \texttt{reverseArray}.}
		\label{fig:met}
	\end{minipage}}
\end{figure*}

\noindent
\textbf{\textit{GINN's Strengths.}}\, To conclude, the graph abstraction method translates to a loop-based abstraction scheme at the source code level. In particular, it sorts programs into hierarchies, where each level presents unique looping constructs for models to learn. Since looping constructs are integral parts of a program, this loop-based abstraction scheme helps GINN to focus on the core of a program for learning its feature representation.
In scalability regard, the graph abstraction method also offers a crucial advantage. Given the abstracted program graphs, even without the manually designed edges, 
GINN mostly deals with small, concise intervals. As a result, information rarely needs to be propagated over long distances, thus significantly enhancing GINN's capability in generalizing from the large graphs, an aspect that the existing GNN struggles with.

\section{Evaluation}
\label{sec:eva}

This section presents an extensive evaluation of GINN on three PL tasks. 
For each task, we compare GINN against the state-of-the-art.



\subsection{Selection Criteria}
To select the right downstream applications for evaluating GINN, we devise a list of criteria
below in the order of importance.
\begin{enumerate}
  \item We prefer tasks for which datasets are publicly accessible and baseline models are open sourced. Furthermore, we only consider models that are built upon Tensorflow, a widely used deep learning library, for reducing the engineering load.
  \item We prefer tasks solved by an influential work that is highly cited in the literature.
  \item We prefer tasks where implementations of GNN-based models are publicly accessible so that we can save our effort in building the baseline with GNN. The reason is
  it's not sufficient to only compare GINN with the state-of-the-art, which GNN may also outperform.
\end{enumerate}

\subsection{Evaluation Tasks}\label{subsec:tasks}
In the end, we select variable misuse prediction task~\cite{allamanis2017learning} (150+ citations) 
in which models aim to infer which variable should be used in a program location. As shown in Figure~\ref{fig:misex}, the variable \texttt{clazz} highlighted in red is misused. Instead, variable \texttt{first} should have been passed as the argument of the function \texttt{Assert.NotNull}. At first,~\citet{allamanis2017learning} use GGNN to predict the misuse variable from AST-based program graphs (more precisely AST with additional edges). Later, the joint model proposed by~\cite{vasic2018neural,Hellendoorn2020Global} has achieved better results and thus becomes the new state-of-the-art.

In addition, we select method name prediction in which models aim to infer the name of a method when given its body. Figure~\ref{fig:met} shows the problem setting where the function name has been stripped for models to predict. The correct answer in this case is \texttt{reverseArray}.~\citet{Alon:2019:CLD:3302515.3290353} (120+ citations) create the first large-scale, cross-project classification task in method names. Later works adopted a generative approach that generates method names as sequences of words~\cite{alon2018code2seq,fernandes2018structured,Wang101145}.


Last, we design a new task to evaluate GINN. Granted, machine learning has inspired a distinctive approach to a board range of problems in program analysis. However, most problems are catered to the strength of machine learning models in mining statistical patterns from data. It's unclear how well they can solve the long-standing problems that are traditionally tackled by symbolic, logic-based methods. In this task, we examine if models can accurately detect null pointer dereference, a type of deep, semantic, and prevalent bugs in software. In particular, we compare GINN against not only the standard GNN but also arguably the state-of-the-art static analysis tool Facebook Infer~\cite{calcagno2015moving,berdine2005smallfoot}.

\ignore{
\begin{table*}
\begin{minipage}{.33\textwidth}
\captionsetup{skip=1pt}
\caption{Result of Enumerative Approach.}
\centering
\adjustbox{max width=1\textwidth}{
\begin{tabular}{c|c|c}
\hline
 \tabincell{c}{Methods} & Accuracy & F1   \\\hline
Baseline & 0.727 & 0.117 \\\hline
GINN & 0.731 & 0.283 \\\hline
\hline
Gain & \multicolumn{2}{|l}{\hspace{1pt}\textbf{0.351} - \textit{0.285} = 0.066} \\\hline
\end{tabular}}
\label{tab:preNPE}
\end{minipage}\hfill
\begin{minipage}{.33\textwidth}
\captionsetup{skip=1pt}
\caption{Result of Joint Model Approach.}
\centering
\adjustbox{max width=1\textwidth}{
\begin{tabular}{c|c|c}
\hline
 \tabincell{c}{Methods} & Accuracy & F1   \\\hline
Baseline & 0.720 & 0.054 \\\hline
GINN & 0.730 & 0.102 \\\hline
\hline
Gain & \multicolumn{2}{|l}{\hspace{1pt}\textbf{0.351} - \textit{0.285} = 0.066} \\\hline
\end{tabular}}
\label{tab:preNPE}
\end{minipage}\hfill
\begin{minipage}{.33\textwidth}
\captionsetup{skip=1pt}
\caption{Result of Method Name Prediction.}
\centering
\adjustbox{max width=1\textwidth}{
\begin{tabular}{c|c|c}
\hline
 \tabincell{c}{Methods} & Precision & F1   \\\hline
Baseline & 0.261 & 0.166 \\\hline
GINN & 0.334 & 0.293 \\\hline
\hline
Gain & \multicolumn{2}{|l}{\hspace{1pt}\textbf{0.351} - \textit{0.285} = 0.066} \\\hline
\end{tabular}}
\label{tab:f1CC}
\end{minipage}\hfill
\end{table*}

\begin{figure*}
\begin{minipage}{.3\textwidth}
\centering
\begin{tikzpicture}[scale=0.92]
  \begin{axis}[
    xmax=100,xmin=0,
    ymin=0,ymax=0.9,
    xlabel=\emph{Graph Size},ylabel=\emph{Accuracy},
    xtick={0,20,40,...,100},
    ytick={0,0.1,0.20,...,0.9},
    xticklabel={\pgfmathparse{\tick}\pgfmathprintnumber{\pgfmathresult}\%},
    ]
     \addplot+ [mark=pentagon*,mark size=2.5pt,every mark/.append style={}] coordinates{(5, 0.566) (15, 0.519) (25, 0.585) (35, 0.537) (45, 0.679) (55, 0.667) (65, 0.585) (75, 0.63) (85, 0.604) (95, 0.648)};
     \addplot coordinates{(5, 0.592) (15, 0.664) (25, 0.763) (35, 0.701) (45, 0.677) (55, 0.835) (65, 0.864) (75, 0.853) (85, 0.712) (95, 0.651)};   
    \legend{\emph{EA},\emph{GINN}}
    \end{axis}
    \end{tikzpicture}
     ~%
    \subcaption{Enumerative Approach}
\label{fig:scaf}
\end{minipage}
\begin{minipage}{.3\textwidth}
\centering
\begin{tikzpicture}[scale=0.92]
  \begin{axis}[
    xmax=100,xmin=0,
    ymin=0,ymax=0.9,
    xlabel=\emph{Graph Size},ylabel=\emph{Accuracy},
    xtick={0,20,40,...,100},
    ytick={0,0.1,0.20,...,0.9},
    xticklabel={\pgfmathparse{\tick}\pgfmathprintnumber{\pgfmathresult}},
  ]
     \addplot+ [mark=pentagon*,mark size=2.5pt,every mark/.append style={}] coordinates{(5, 0.589)  (15, 0.614)  (25, 0.665) (35, 0.761)  (45, 0.734)  (55, 0.819)  (65, 0.886)  (75, 0.795)  (85, 0.795)  (95, 0.657)};
     \addplot coordinates{(5, 0.602)  (15, 0.629)  (25, 0.753) (35, 0.67)  (45, 0.745)  (55, 0.824)  (65, 0.827)  (75, 0.853)  (85, 0.701)  (95, 0.698)};
     \legend{\emph{Baseline}, \emph{GINN}}
    \end{axis}
    \end{tikzpicture}%
    ~%
\subcaption{Joint Model}
\label{fig:ttf}
\end{minipage}
\begin{minipage}{.3\textwidth}
\centering
\begin{tikzpicture}[scale=0.92]
  \begin{axis}[
    xmax=100,xmin=0,
    ymin=0,ymax=0.9,
    xlabel=\emph{Graph Size},ylabel=\emph{Accuracy},
    xtick={0,20,40,...,100},
    ytick={0,0.1,0.20,...,0.9},
    xticklabel={\pgfmathparse{\tick}\pgfmathprintnumber{\pgfmathresult}},
  ]
     \addplot+ [mark=pentagon*,mark size=2.5pt,every mark/.append style={}] coordinates{(5, 0.589)  (15, 0.614)  (25, 0.665) (35, 0.761)  (45, 0.734)  (55, 0.819)  (65, 0.886)  (75, 0.795)  (85, 0.795)  (95, 0.657)};
     \addplot coordinates{(5, 0.602)  (15, 0.629)  (25, 0.753) (35, 0.67)  (45, 0.745)  (55, 0.824)  (65, 0.827)  (75, 0.853)  (85, 0.701)  (95, 0.698)};
     \legend{\emph{Baseline}, \emph{GINN}}
    \end{axis}
    \end{tikzpicture}%
    ~%
\subcaption{Method Naming}
\label{fig:ttf2}
\end{minipage}
\caption{Experiment results.}
\end{figure*}
}

\subsection{Variable Misuse Prediction}
The state-of-the-art model for predicting variable misuse bugs is proposed by~\citet{Hellendoorn2020Global}. Their model is an instantiation of the conceptual framework invented by~\citet{vasic2018neural}. In this section, we first describe this framework, then we explain how to instantiate it using GINN.

\subsubsection{Joint Model for Localization and Repair}
\label{subsubsec:joint}
\citet{vasic2018neural} propose a joint model structure that learns to localize and repair variable misuse bugs simultaneously. Intuitively, their design exploits an important property of variable misuse bugs. That is both the misuse (\ie incorrectly used) and the repair (\ie should have been used) variable have already been defined somewhere in the program, therefore, the key is to identify their locations.~\citet{vasic2018neural} learned probability distributions over the sequence of tokens of a buggy program from which they pick the token of the highest probability to be the misuse or repair variable.
At a high-level, the joint model consists of three major components: initial embedding layer, core model, 
and two separate classifiers.
Below we describe each component in detail.

\vspace*{3pt}
\noindent
\textbf{\textit{Initial Embedding Layer.}}\,
Initial embedding layer is responsible for turning each symbol from the input vocabulary (\eg a token or a type of AST node) into a numerical vector, the format that is amenable to deep neural networks. A simple, crude method is to use the one-hot vectors, a $N \times N$ matrix representing every word in a vocabulary ($N$ denotes the size of the vocabulary). Each row corresponds to a word, which consists of 0s in all cells with the exception of a single 1 in a cell used uniquely to identify the word. For example, given a vocabulary, \{\texttt{a},\,\,\texttt{+},\,\,\texttt{b},\,\,\texttt{binary-exp}\}, the matrix $M$ in Equation~\ref{equ:one-hot} depicts the one-hot vectors of each token. A drawback of one-hot vectors is they don't capture the semantic relationship of words in the vocabulary as vectors are uniformly scattered in a high dimensional space (\ie the Euclidean distance between any vector with every other vector is a constant). A common remedy is to have the one-hot vectors multiply another matrix $W_{1} \in \mathbb{R}^{N\times d}$ to produce $\mathcal{E} \in \mathbb{R}^{N\times d}$.
The intuition is to expand the network's capacity with more learnable parameters in $W_{1}$ for searching a precise embedding matrix $\mathcal{E}$ for each word in the vocabulary.
$d$ denotes the number of columns in $W_{1}$, which is also the size of the embedding vectors. Equation~\ref{equ:one-hot} gives an example where $d$ equals to 3. The values of $W_{1}$ is initialized randomly and will be learned simultaneously with the network during training.

\vspace*{3pt}
\noindent
\textbf{\textit{Core Model.}}\, Given the embedding vector of each token, the core model computes the numerical representation of the input program, expressed as another matrix $\mathcal{C}$ where each row represents a token. Many neural architectures that can play the role of core model.~\citet{vasic2018neural} used Long Short Term Memory networks (LSTM)~\cite{LST}, a specialization of Recurrent Neural Networks (RNN).~\citet{Hellendoorn2020Global} explored other alternatives including GGNN, Transformer~\cite{vaswani2017attention}, the state-of-the-art sequence model, and RNN Sandwich, a mixture of sequence and graph models, which achieves the state-of-the-art results. For now, a core model can be considered as a black-box with an abstract interface $\phi: \mathcal{E}\rightarrow\mathcal{C}$. We defer the discussion of two instantiations of the core model to later sections.

\begin{equation}
\begin{gathered}
  \includegraphics[height=3cm]{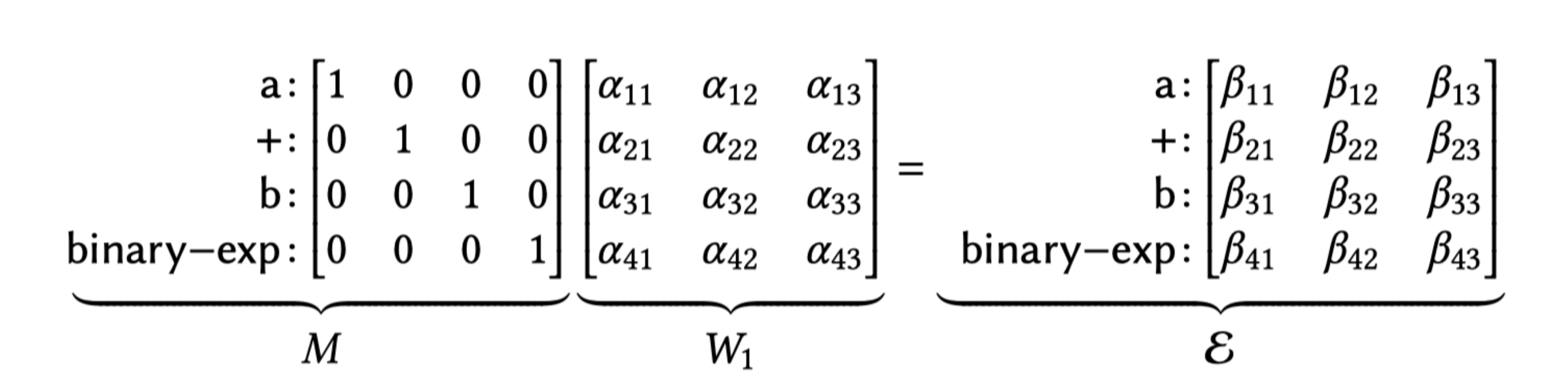}
  \end{gathered} 
\label{equ:one-hot}
\end{equation}

\vspace*{3pt}
\noindent
\textbf{\textit{Classifiers for Misuse and Repair Variables.}}\, This layer is responsible for predicting the location of misuse and repair variable\footnote{We followed~\citet{Hellendoorn2020Global}'s approach to consider a single variable misuse bug per program.}.
It distributes two probabilities to each token in the program: one for identification of the misuse variable and the other for the repair variable. The token with the highest probability of being the misuse or repair variable will be picked as the output.
At a technical level,~\citet{vasic2018neural} performed a linear transformation on matrix $\mathcal{C} \in \mathbb{R}^{k \times h}$ using another matrix $W_{2} \in \mathbb{R}^{h \times 2}$ to produce $P_{val} \in \mathbb{R}^{k \times 2}$. As explained before, $k$ is the number of tokens in a program; $h$ is the size of the embedding vector produced by a core model.
\begin{equation*}\label{equ:lin}
    P_{val} = \mathcal{C}W_{2}
\end{equation*}

Since $P_{val}$ is a non-normalized output, we apply $\mathit{softmax}$ function, a common practice in machine learning,
to map $P_{val}$ into probability distributions over the token sequence of an input program. Note that $\mathit{axis}=1$ in Equation~\ref{equ:applysoft} means the normalization happens along each column in which each value denotes the probability of one token being the misuse or repair variable.
\begin{equation}\label{equ:applysoft}
    P_{\mathit{pro}} = \mathit{softmax}(P_{val},\mathit{axis}=1)
\end{equation}

Finally, the model is trained to nudge $P_{\mathit{pro}}$ as close to the true distribution of misuse and repair variables as possible. In particular, we minimize the cross-entropy loss expressed by the difference between $P_{\mathit{pro}}$ and the true distributions, denoted by $\mathit{Loc}$ and $\mathit{Rep}$.
Similar to the one-hot vectors,  $\mathit{Loc}$ (\resp $\mathit{Rep}$) assigns a value of 1 to the actual index of misuse (\resp repair) variable among all tokens of the input program and 0 otherwise. We give the details of the loss function in Appendix~\ref{app:loss}.

\subsubsection{Instantiations of the Core Model}\label{subsubsec:ins} In this section, we review two instantiations of the core model presented in~\cite{Hellendoorn2020Global}: GGNN and RNN Sandwich. In particular, we discuss how each of them computes the matrix $\mathcal{C}$ introduced earlier.

As explained in Section~\ref{subsec:tasks},
GGNN works with AST as the backbone of graphs. To collect the numerical representations of each token in the program, ~\citet{Hellendoorn2020Global} extracted states of terminals nodes after GGNN had completed the message-passing routine, and then stacked them into 
matrix $\mathcal{C}$ (Equation~\ref{equ:GGNNC}).
\begin{equation}
\label{equ:GGNNC}
\mathcal{C} = \mathit{stack}([\mu_{v_{1}},\dots,\mu_{v_{n}}],\mathit{axis}=0), \forall v_{i} \in \Sigma
\end{equation}
where $\Sigma$ denotes the set of terminal nodes in AST; $\mathit{axis}=0$ indicates each $\mu_v$ forms a row in $\mathcal{C}$.

Next, we discuss RNN Sandwich~\cite{Hellendoorn2020Global}, the state-of-the-art model in predicting variable misuse bugs. In a nutshell, RNN Sandwich adopts a hybrid neural architecture combining a sequence model RNN and a graph model GGNN
intending to 
get the best of both worlds. That is, not only exploiting the semantic code structure using GGNN but also streamlining information flow through token sequences using RNN. Technically, the way it works is the following: (1) first, the matrix $\mathcal{E}$ produced in the initial embedding layer will be fed into an RNN. The results are a list of vectors, $h_1,\dots,h_n$, each of which represents a token. Compared to the initial embedding matrix $\mathcal{E}$, $h_1,\dots,h_n$ increases the precision of the token representation by capturing the temporal properties they display in a sequence.
Equation~\ref{equ:rnn} computes $h_1,\dots,h_n$. In simplest terms, RNN takes two vectors at each time step --- the embedding vector of the $t$-th token in the sequence, denoted by $\mathcal{E}[t]$ (assuming the $t$-th row in $\mathcal{E}$ is the embedding of the $t$-th token), and RNN's current hidden state after consuming the embedding of the previous token. The output is 
\makebox[\linewidth][s]{the new hidden state $h_{t}$. For interested readers, Appendix~\ref{app:RNN} provides a more detailed description}
\begin{equation}
\label{equ:rnn}
h_{t} = \mathit{RNN}(\mathcal{E}[t], h_{t-1})
\end{equation}
of RNN's computation model. (2) Next, GGNN takes over and computes the state vector for each node via the message-passing protocol. For node initialization, the token representations computed in Step (1) are assigned to the terminal nodes in an AST (Equation~\ref{equ:assign}) while the non-terminal nodes
\begin{equation}
\label{equ:assign}
\mu_{v_{i}} = h_{i}, \forall v_{i} \in \Sigma, \forall h_{i} \in [h_1,\dots,h_n]
\end{equation}
keep their representations computed from the initial embedding layer. (3) Finally, after GGNN has computed the state of each node in an AST, RNN takes back those of the terminal nodes and computes a new representation for each token, $h_{1}^\prime,\dots,h_{n}^\prime$, according to Equation~\ref{equ:rnn2} where $\mu_{v_{t}}$ 
\makebox[\linewidth][s]{denotes the state vector of the node corresponding to the $t$-th token in the program. Finally, the}
\begin{equation}
\label{equ:rnn2}
h_{t}^\prime = \mathit{RNN}(\mu_{v_{t}}, h_{t-1}^\prime)
\end{equation}
matrix $\mathcal{C}$ can be computed by Equation~\ref{equ:rnnsan}. 
\begin{equation}
\label{equ:rnnsan}
    \mathcal{C} = \mathit{stack}([h_{1}^\prime,\dots,h_{n}^\prime],\mathit{axis}=0)
\end{equation}

\vspace{.5pt}
\subsubsection{New Instantiations of Core Models Using GINN}~\label{subsubsec:ginnins} We propose two new instantiations built upon GINN to pair with GGNN and RNN Sandwich. In order to adapt GINN as an instantiation of the core model, an important issue needs to be addressed. Recall the interface defined for the core model: $\phi: \mathcal{E}\rightarrow\mathcal{C}$, since GINN can not compute representations of tokens from a standard control flow graph, matrix $\mathcal{C}$ can't be produced. Therefore, we modify control flow graphs as follows: (1) first we split a graph node representing a basic block into multiple nodes, each of which represents a single statement. Subsequently, we add additional edges to connect every statement with its immediate successor within the same basic block. For edges on the original control flow graphs, we change their start (\resp end) nodes from a basic block to its last (\resp first) statement after the split; (2) next, we replace each statement node with a sequence of nodes, each of which represents a token of the statement. Specifically, every token node is connected to its immediate successor, and the first token node will become the new start or end node of edges that were connecting the statement nodes before.

Figure~\ref{fig:cfgex} depicts an example of the new graph representation given the function in Figure~\ref{fig:codeex}. Note that our modification will not alter the partitioning of a graph into intervals. Instead, we only replace nodes denoting basic blocks with those denoting token as the new constituents of an interval. We have included a detailed explanation in Appendix~\ref{app:parti} for readers’ perusal. Under the new graph representation, GINN can compute the state vector for every token in the program, and in turn the matrix $\mathcal{C}$. Therefore, we can now instantiate the core model with GINN. Similarly, we can easily construct a variant of RNN Sandwich by swapping out GGNN for GINN as the new graph component.
\begin{figure}[thb!]
    \captionsetup[subfigure]{aboveskip=-1pt,belowskip=-1pt}
    \centering
    \begin{subfigure}{0.4\textwidth}
\begin{lstlisting}[linewidth=5.5cm,morekeywords={bool}]
static bool FindIndexOfMaxElement(int[] arr)
{              
    int max = int.MinValue, idx = 0;

    for (int i = 0; i < arr.Length; i++)
    {
        if (arr[i] > max)
        {
            max = arr[i];
            idx = i;
        }
    }

    return idx;
}

            
\end{lstlisting}        
    	\caption{}
    	\label{fig:codeex}
    \end{subfigure}
    \begin{subfigure}{0.421\textwidth}
        \centering
        \includegraphics[width=\textwidth]{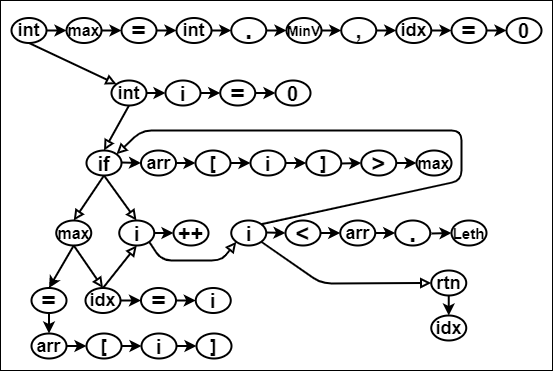}
    	\caption{}
        \label{fig:cfgex}
    \end{subfigure}
    \setlength{\belowcaptionskip}{-4pt}
    \caption{The program graph in~(\subref{fig:cfgex}) that represents the function in~(\subref{fig:codeex}). Edges with hallow arrow heads represent the control flows between program statements. The other edges connect tokens in a sequence.}
    \label{fig:gr}
\end{figure}

\subsubsection{Experimentation} Given the conceptual framework of the joint model and various instantiations discussed above, we now describe the experiment setup. 

\vspace*{3pt}
\noindent
\textbf{\textit{Dataset, Metric, and Baseline.}}\, \citet{Hellendoorn2020Global} worked with ETH Py150~\cite{Raychev2984041}, a publicly accessible dataset of python programs, on top of which they introduced buggy programs to suit the task of variable misuse prediction. Since the actual dataset used in their experiment is not publicly available, we follow their pre-processing steps in an attempt to replicate their dataset. Major tasks include the deduplication of ETH Py150, generation of buggy programs, and preservation of original programs as correct examples. Like their dataset, we maintain a balance between buggy and correct programs. In the end, we have 3M programs in total, among which We use 2M for training, 245K for validation, and 755K for test. A detailed description of the data generation is provided in Appendix~\ref{app:dataset}. Similarly, we adopt the metrics proposed in~\cite{vasic2018neural} to measure the model performance, such as (1) classification accuracy, the percentage of programs in the test set that are correctly classified as either buggy or correct; (2) localization accuracy, the percentage of buggy programs for which the bug location is correctly predicted; and (3) localization+repair accuracy (or joint accuracy), the percentage of buggy programs for which both the location and repair are correctly predicted. Regarding baselines, we use GGNN and RNN Sandwich as the two instantiations of the core model.

\vspace*{3pt}
\noindent
\textbf{\textit{Implementation.}}\, Since ETH Py150 provides AST for programs in the dataset, we convert each AST into the graph representation that GINN consumes (Appendix~\ref{app:convert}). We also implement Algorithm~\ref{alg:intervals} to partition the graphs into intervals. Since prior works have not open-sourced the joint model framework, we implemented our own in Tensorflow.
Regarding the core models, we use~\citet{allamanis2017learning}'s implementation of GGNN,
on top of which we realize GINN's graph abstraction method while keeping all GGNN's default parameters (\eg size of node embeddings, number of layers, \etc) intact.
We only implement one abstraction cycle of GINN's (original graph$\rightarrow$highest order graphs$\rightarrow$original graphs) as more cycles do not make a significant difference. For RNN Sandwich, we use the large sandwich model~\cite{Hellendoorn2020Global}, which is shown to work better than the vanilla model presented in Section~\ref{subsubsec:ins}. The improvement is due to the increased frequency of the alternation between sequence and graph models. That is, instead of inserting RNN only before and after GNN's computation, they wrap every few rounds of message-passing inside of GNN with an RNN to facilitate a thorough exchange between two neural architectures. We couple GINN (\resp GGNN) with an RNN to implement the GINN- (\resp GGNN-) powered RNN Sandwich.
Note that~\citet{Hellendoorn2020Global} also propose a lightweight graph representation by keeping only the terminal nodes from an AST as the input for both GGNN and RNN Sandwich. Even though these new program graphs have made models faster to train, their accuracy often decreases due to the lesser information contained in the input graphs. Hence, we don't consider this lightweight graphs in our evaluation.     


\begin{table}[thb!]
\centering
\caption{Results for the different instantiations of core models. For GINN and GINN-powered RNN Sandwich, we also include their results of an ablation study and an evaluation on an alternative design.}
\adjustbox{max width=1\columnwidth}{
\begin{tabular}{r|c|c|c|c}
\hline
Models & Configuration & \tabincell{c}{Classification \\Accuracy} & \tabincell{c}{Localization \\Accuracy} & \tabincell{c}{Localization+Repair \\Accuracy} \\\hline\hline
GGNN & Original & {74.0}\% & {58.1}\% & {56.0}\%    \\\hline 
\multirow{3}{*}{\textbf{GINN}} & Original & \textbf{74.9\%} & \textbf{60.5\%} & \textbf{60.0\%} \\\cline{2-5}
                                & Ablated  & 74.7\% & 58.8\% &57.6\% \\\cline{2-5}
                                & Alternative & 74.0\% & 59.5\% & 58.7\% \\\hline
Sandwich (GGNN) & Original & 74.7\% & 62.9\% & 61.1\% \\\hline
\multirow{3}{*}{\textbf{Sandwich (GINN)}} & Original & \textbf{75.8\%} & \textbf{69.3\%} & \textbf{68.4\%} \\\cline{2-5}
                                              & Ablated  & 75.1\% & 62.4\% & 62.3\% \\\cline{2-5}
                                              & Alternative & 75.0\% & 63.1\% & 62.0\%\\\hline
\end{tabular}
}
\label{tab:res}
\end{table}

All experiments including those to be presented later are performed on a desktop that runs
Ubuntu 16.04 having 3.7GHz i7-8700K CPU, 32GB RAM, and NVIDIA GTX 1080 GPU. As a pre-test, we repeat the experiment presented in~\cite{Hellendoorn2020Global}, and observed comparable model performance for both GGNN and GGNN-based RNN Sandwich, signaling the validity of both our model implementation and data curation. Interested readers may refer to Appendix~\ref{app:rep} for details.

\subsubsection{Results}\label{subsubsec:res}
First, we compare the performance of GINN (\resp GINN-powered RNN Sandwich) against GGNN (\resp GGNN-powered RNN Sandwich). Then we investigate the scalability of these four neural architectures. Later, we conduct ablation studies to gain a deeper understanding of GINN's inner workings. Finally, we evaluate an alternative design of GINN.

\vspace*{3pt}
\noindent
\textbf{\textit{Accuracy.}}\,
Table~\ref{tab:res} depicts the results of all four models using the aforementioned metrics (Appendix~\ref{app:tt} shows their training time). We focus on rows for the original configuration of each model, and discuss the rest later. Regarding the classification accuracy, all models perform reasonably well, and GINN-powered RNN Sandwich and GINN are the top two core models despite by a small margin. For more challenging tasks like localization or repair, GINN-powered RNN Sandwich now displays a significant advantage over all other models, in particular, it outperforms GGNN-powered RNN Sandwich, the state-of-the-art model in variable misuse prediction, by 6.4\% in localization and 7.3\% in joint accuracy.
To dig deeper, We manually inspect the predictions made by each model and find that GINN-powered RNN Sandwich is considerably more precise at reasoning the semantics of a program. Below we gave two examples to illustrate our findings. 

For Figure~\ref{fig:varmis1}, GINN-powered RNN Sandwich is the only model that not only locates but also fixes the misused variable (\ie \texttt{orig\_sys\_path} highlighted within the shadow box). In contrast, all baseline models consider \texttt{orig\_sys\_path} as the correctly used variable, which is not totally unreasonable. In fact, appending an item (\texttt{item}) when it is not yet in the list (\texttt{orig\_sys\_path}) is a very common pattern models need to learn. However, in this case,
if \texttt{orig\_sys\_path} was the correct variable, \texttt{new\_sys\_path} would have been an empty list when being assigned to \texttt{sys.path[:0]} in the last line. GINN-powered RNN Sandwich is capable of capturing the nuance of the semantics this program denotes, and not getting trapped by the common programming paradigms. As for Figure~\ref{fig:varmis2}, GINN and GGNN-powered RNN Sandwich also correctly localizes the misused variable \texttt{request}, but none of them predicts the right repair variable \texttt{sq}. The signal there is the fact that \texttt{fts} is a list, which does not have \texttt{keys()} method. GINN-powered RNN Sandwich is again the only neural architecture that produces the correct prediction end-to-end, demonstrating its higher precision in reasoning the semantics of a program.

\begin{figure}[thb!]
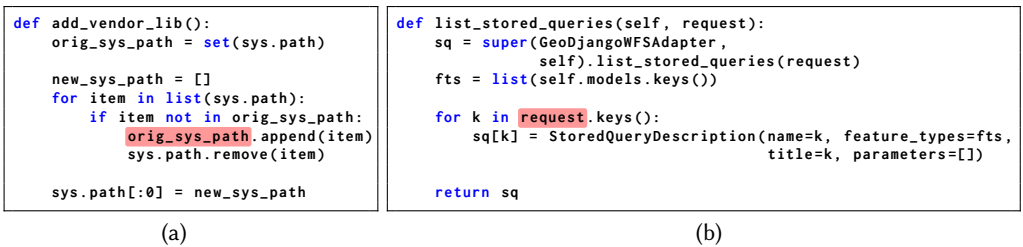

    \captionsetup[subfigure]{aboveskip=-1pt,belowskip=-4pt}
    \centering
    \begin{subfigure}{0.31\textwidth}
    	\lstset{style=mystyle}
\begin{lstlisting}[linewidth=4.75cm,morekeywords={def,list,set,not,in}]
def add_vendor_lib():
    orig_sys_path = set(sys.path)

    new_sys_path = []
    for item in list(sys.path):
        if item not in orig_sys_path:
            @orig_sys_path@.append(item)
            sys.path.remove(item)

    sys.path[:0] = new_sys_path
\end{lstlisting}        
    	\caption{}
    	\label{fig:varmis1}
    \end{subfigure}
    \,\,\,
    \begin{subfigure}{0.6\textwidth}
        \centering
\begin{lstlisting}[linewidth=8.2cm,morekeywords={def,super,list,in}]
def list_stored_queries(self, request):
    sq = super(GeoDjangoWFSAdapter, 
               self).list_stored_queries(request)
    fts = list(self.models.keys())

    for k in @request@.keys():
        sq[k] = StoredQueryDescription(name=k, feature_types=fts, 
                                       title=k, parameters=[])
    
    return sq
\end{lstlisting}        
    	\caption{}
    	\label{fig:varmis2}
    \end{subfigure}
    \setlength{\belowcaptionskip}{-18pt}
    \caption{Two programs only GINN-powered RNN Sandwich makes the correct prediction end-to-end.}
    \label{fig:misexs}
\end{figure}

\vspace*{3pt}
\noindent
\textbf{\textit{Scalability.}}\, We further analyze the results we obtained from the previous experiment to investigate the scalability of the four neural architectures. In particular, we divide the entire test set into ten subsets, each of which consists of graphs of similar size. We then record the performance of all models on each subset starting from the smallest to the largest graphs. Note that GGNN, at the core of both compared baselines, has already had additional edges incorporated into the program graphs to alleviate the scalability concern~\cite{allamanis2017learning}. However, as explained in Section~\ref{sec:intro},~\citet{allamanis2017learning} do not provide a principled guideline as to where exactly to add the edges given a program graph specifically. Therefore, we can only follow their approach by universally adding all types of edges they proposed (9 in total which we list in Appendix~\ref{app:edges}) to all program graphs. In contrast, we do not incorporate any additional edge apart from those that are present in the graph representation for GINN. 

\setlength{\intextsep}{12pt}

\begin{figure*}[thb!]
\centering
   \begin{tikzpicture}[scale=0.72]
   \begin{axis}[
    xmax=10,xmin=1,
    ymin=0.55,ymax=1.05,
    xlabel=\emph{N-th Set (Ranked by Size of Graphs)},ylabel=\emph{Classification Accuracy},
    xtick={1,2,3,...,10},
    ytick={0.55,0.6,...,1.05},
    xticklabels={1,2,...,10},
   ]
     \addplot+ [mark=pentagon*,mark size=2.5pt] coordinates{(1, 0.739) (2, 0.741) (3, 0.740) (4, 0.741) (5, 0.740) (6, 0.739) (7, 0.741) (8, 0.740) (9, 0.741) (10, 0.741)};
     \addplot+ coordinates{(1, 0.744) (2, 0.716) (3, 0.749) (4, 0.748) (5, 0.749) (6, 0.750) (7, 0.746) (8, 0.749) (9, 0.765) (10, 0.782)};
     \addplot+ coordinates{(1, 0.689) (2, 0.766) (3, 0.765) (4, 0.767) (5, 0.711) (6, 0.766) (7, 0.766) (8, 0.689) (9, 0.766) (10, 0.782)};
     \addplot+ [mark=triangle*,mark size=2.5pt] coordinates{(1, 0.692) (2, 0.769) (3, 0.769) (4, 0.769) (5, 0.714) (6, 0.769) (7, 0.846) (8, 0.692) (9, 0.769) (10, 0.786)};
    \legend{GGNN, GINN, Sand, I-Sand}
    \end{axis}
    \end{tikzpicture}%
    ~%
    \begin{tikzpicture}[scale=0.72]
   \begin{axis}[
    xmax=10,xmin=1,
    ymin=0.4,ymax=1.2,
    xlabel=\emph{N-th Set (Ranked by Size of Graphs)},ylabel=\emph{Location Accuracy},
    xtick={1,2,...,10},
    ytick={0.4,0.5,...,1.2},
    xticklabels={1,2,...,10},
    ]
     \addplot+ [mark=pentagon*,mark size=2.5pt,every mark/.append style={}] coordinates{(1, 0.433+0.3) (2, 0.363+0.3) (3, 0.341+0.3) (4, 0.277+0.3) (5, 0.254+0.3) (6, 0.269+0.3) (7, 0.270+0.3) (8, 0.219+0.3) (9, 0.291+0.3) (10, 0.294+0.3)};
     \addplot coordinates{(1, 0.266+0.3) (2, 0.583+0.3) (3, 0.498+0.3) (4, 0.376+0.3) (5, 0.242+0.3) (6, 0.260+0.3) (7, 0.303+0.3) (8, 0.253+0.3) (9, 0.278+0.3) (10, 0.303+0.3)};   
     \addplot coordinates{(1, 0.500+0.3) (2, 0.420+0.3) (3, 0.385+0.3) (4, 0.304+0.3) (5, 0.291+0.3) (6, 0.298+0.3) (7, 0.336+0.3) (8, 0.275+0.3) (9, 0.340+0.3) (10, 0.365+0.3)};   
     \addplot+ [mark=triangle*,mark size=2.5pt] coordinates{(1, 0.407+0.3) (2, 0.604+0.3) (3, 0.604+0.3) (4, 0.435+0.3) (5, 0.281+0.3) (6, 0.323+0.3) (7, 0.421+0.3) (8, 0.337+0.3) (9, 0.359+0.3) (10, 0.367+0.3)}; 
    \legend{GGNN, GINN, Sand, I-Sand}
    \end{axis}
    \end{tikzpicture}%
    ~%
     \begin{tikzpicture}[scale=0.72]
   \begin{axis}[
    xmax=10,xmin=1,
    ymin=0.3,ymax=1.1,
    xlabel=\emph{N-th Set (Ranked by Size of Graphs)},ylabel=\emph{Location+Repair Accuracy},
    xtick={1,2,...,10},
    ytick={0.3,0.4,...,1.1},
    xticklabels={1,2,...,10},
    ]
    
     \addplot+ [mark=pentagon*,mark size=2.5pt,every mark/.append style={}]     coordinates{(1, 0.463+0.28) (2, 0.350+0.28) (3, 0.310+0.28) (4, 0.307+0.28) (5, 0.278+0.28) (6, 0.261+0.28)   (7, 0.250+0.28) (8, 0.220+0.28) (9, 0.225+0.28) (10, 0.239+0.28)};
     \addplot coordinates{(1, 0.282+0.28) (2, 0.579+0.28) (3, 0.506+0.28) (4, 0.335+0.28) (5, 0.250+0.28) (6, 0.240+0.28) (7, 0.283+0.28) (8,.250+0.28) (9, 0.268+0.28) (10, 0.230+0.28)};
     \addplot coordinates{(1, 0.490+0.28) (2, 0.410+0.28) (3, 0.375+0.28) (4, 0.290+0.28) (5, 0.281+0.28) (6, 0.286+0.28) (7, 0.326+0.28) (8, 0.263+0.28) (9, 0.329+0.28) (10, 0.355+0.28)};
     \addplot+ [mark=triangle*,mark size=2.5pt] coordinates{(1, 0.401+0.28) (2, 0.571+0.28) (3, 0.510+0.28) (4, 0.420+0.28) (5, 0.350+0.28) (6, 0.353+0.28) (7, 0.373+0.28) (8, 0.289+0.28) (9, 0.380+0.28) (10, 0.407+0.28)}; 
    \legend{GGNN, GINN, Sand, I-Sand}
     \end{axis} 
     \end{tikzpicture}
    \caption{Investigating the scalability of all four neural architectures. I-Sand (\resp Sand) denotes the GINN-powered RNN Sandwich (\resp GGNN-powered RNN Sandwich).}
    \label{fig:scavar}
\end{figure*}
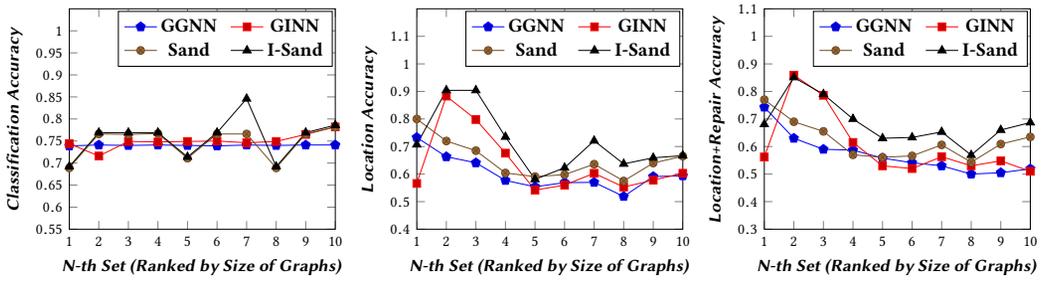

Figure~\ref{fig:scavar} shows how the performance of each model varies with the size of graphs. We skip the metric of classification accuracy under which models are largely indistinguishable. For localization and joint accuracy, we make several observations of the performance trend for each model. First, GGNN suffers the largest performance drop among all models --- more than 10\% (\resp 20\%) under location (\resp joint) accuracy. This phenomenon shows there is indeed a scalability issue with GGNN. Second, the combination of sequence and graph model helps to make the neural architecture more scalable than GGNN alone. Furthermore, by replacing GGNN with GINN in RNN Sandwich, we obtain a model that is clearly the most scalable. Barring few sets of graphs, GINN-powered RNN Sandwich is considerably more accurate than any other model, especially under the joint accuracy.
Finally, 
GINN alone also manages to improve the scalability of GGNN. As described earlier, the improvement is magnified under the assistance of the RNN in the sandwich model. 

\vspace*{3pt}
\noindent
\textbf{\textit{Ablation Study.}}\,
Unlike the prior work~\cite{allamanis2017learning}, all edges in our graph representation are indispensable, thus can not be ablated. Therefore, the goal of this ablation study is to quantify the influence of the graph abstraction method, the crux of GINN's learning approach, on GINN's performance. Because GINN is identical to GGNN without the graph abstraction operators, we evaluate GGNN when GINN's program graphs --- a variant of control flow graph --- are provided as input data. Rows for ablated configuration in  Table~\ref{tab:res} show GINN (\resp GINN-powered RNN Sandwich) becomes only slightly more accurate than GGNN (\resp GGNN-powered RNN Sandwich) after the ablation, in other words, our graph abstraction method indeed helps models to learn.

\noindent
\textbf{\textit{Alternative Design.}}\, As mentioned in Section~\ref{sec:met}, we evaluate an alternative design of GINN in which freshly 
created nodes --- due to the heightening operator --- are initialized to be the average of the replaced nodes; conversely, for nodes that are created by the lowering operator, they receive an equal share of the split node that they emerge from. In essence, we assign the weight, $\alpha_{v}$, used in Equation~\ref{equ:merge} and~\ref{equ:split} to be $1/\lvert I^{n}_h \rvert$. Rows of alternative configuration in Table~\ref{tab:res} show that the new design results in a notable decrease in model accuracy, which motivates the original design of GINN.

\subsection{Method Name Prediction}
The state-of-the-art model in method name prediction, sequence GNN~\cite{fernandes2018structured}, adopts a typical encoder-decoder architecture~\cite{devlin2014fast,cho-cho2014learning}.\footnote{\citet{Wang101145} proposed a model to solve the same problem. Since (1) they did not perform a head-to-head comparison against sequence GNN; and (2) their methodology requires program executions, we do not consider their model for this task.} In a similar vein to RNN Sandwich, sequence GNN also employs a combination of graph and sequence model to encode a program, however, the synergy between the two neural architectures in this case is considerably simpler. They only perform the first two steps of the vanilla joint model presented in Section~\ref{subsubsec:ins}. That is, RNN first learns the sequence representation of each token in a program before GGNN computes the state for every node in the AST.
For decoder, they use another RNN, which generates the method name as a sequence of words. While decoding, it also attends to the states of nodes in the AST, a common technique for improving the accuracy of models adopting encoder-decoder architecture~\cite{bahdanau2014neural}. 

Instantiating the framework of sequence GNN's with GINN is a straightforward task. We use GINN as the new graph model for the encoder while keeping the remaining parts of the framework intact. Like prior experiments, GINN works with the same graph format derived from control flow graphs and customized to include nodes of tokens.  

\subsubsection{Experimentation} Given the sequence GNN's framework and the instantiation with GINN, we describe the setup of our experiment.

\vspace*{3pt}
\noindent
\textbf{\textit{Dataset, Metric, and Baseline.}}\, We consider Java-small proposed by~\citet{alon2018code2seq}, a publicly available dataset used in~\cite{fernandes2018structured}, adopting the same train-validation-test splits they have defined. We also follow~\citet{fernandes2018structured}'s approach to measure performance using F1 and ROGUE score over the generated words (Appendix~\ref{app:mnmetrics}). For sequence GNN baseline, we use the model that achieves state-of-the-art results on Java-small, which employs bidirectional-LSTM as the sequence and GGNN as the graph model for encoder and another LSTM for decoder. Apart from the aforementioned attention mechanism, decoder also integrates a pointer network~\cite{NIPS2015_5866} to directly copy tokens from the input program. 

\vspace*{3pt}
\noindent
\textbf{\textit{Implementation.}}\, 
We use JavaParser to extract AST out of programs in Java-small. We adopt the same procedure as before to convert AST into the graph representation that GINN consumes. Regarding the model implementation, we use the model code of sequence GNN open-sourced on GitHub, within which we implemented GINN's graph abstraction methods. Like before, all default model parameters in their implementation are kept as they are and we only implement one abstraction cycle within GINN. We name our model built out of GINN sequence GINN.

\subsubsection{Results}\label{subsubsec:metres}
In this section, we first compare the performance of sequence GINN against sequence GNN. Then we investigate the scalability of both neural architectures. Finally, we conduct a similar ablation study and alternative design evaluation with those presented in the previous task. 

\begin{figure}[t]
    \adjustbox{max width=0.4\textwidth}{
    \begin{subfigure}{0.5\textwidth}
    	\lstset{style=mystyle}
\begin{lstlisting}[linewidth=7.2cm]
private int @extractModifiers@(GroovySourceAST ast) {    
	GroovySourceAST modifiers = ast.childOfType(MODIFIERS);

	if (modifiers == null) return 0;
	int modifierFlags = 0;    
	for (GroovySourceAST child = 
			(GroovySourceAST) modifiers.getFirstChild(); 
			child != null; 
			child = (GroovySourceAST) child.getNextSibling()) {            
		switch (child.getType()) {        
		case LITERAL_private:        
			modifierFlags |= Modifier.PRIVATE; break;
		case LITERAL_protected:
			modifierFlags |= Modifier.PROTECTED; break;
		case LITERAL_public:
			modifierFlags |= Modifier.PUBLIC; break;
		case FINAL:
			modifierFlags |= Modifier.FINAL; break;            
		case LITERAL_static:
			modifierFlags |= Modifier.STATIC; break;            
	}}
    
	return modifierFlags;   
}
\end{lstlisting}        
    	\caption{Predicted to be {\footnotesize\texttt{SelectModifier}} by the baseline.}
    	\label{fig:met1}
    \end{subfigure}}    
    \,
    \adjustbox{max width=0.58\textwidth}{
    \begin{subfigure}{0.6\textwidth}
    \vspace{3.3pt}
        \centering
    	\lstset{style=mystyle}
\begin{lstlisting}[basicstyle=\linespread{1.09}\scriptsize\ttfamily\bfseries,linewidth=9.8cm]
public void @applyRules@(%\color{black}{@}%Nullable String pluginId, Class<?> clazz) {
	ModelRegistry modelRegistry = target.getModelRegistry();    
	
	Iterable<Class<? extends RuleSource>> declaredSources = 
			ruleDetector.getDeclaredSources(clazz);
    
	for (Class<? extends RuleSource> ruleSource : declaredSources) {    
		Iterable<ExtractedModelRule> rules = 
				ruleInspector.extract(ruleSource).getRules();
		for (ExtractedModelRule rule : rules) {        			
			for (Class<?> dependency : rule.getRuleDependencies()) {            
				target.getPluginManager().apply(dependency);                
			}
			
			rule.apply(modelRegistry, ModelPath.ROOT);            
		}
	}
}
\end{lstlisting}        
    	\caption{Predicted to be {\footnotesize\texttt{CheckSource}} by the baseline.}
    	\label{fig:met2}
    \end{subfigure}}
    \setlength{\abovecaptionskip}{7pt}    
    \setlength{\belowcaptionskip}{-4pt}    
    \caption{Two programs sequence GINN predicts correctly but not the baseline.}
    \label{fig:metexs}
\end{figure}

\vspace*{3pt}
\noindent
\textbf{\textit{Accuracy.}}\, Table~\ref{tab:resmet} depicts the results of the two models. Sequence GINN significantly outperforms the sequence GNN across all metrics (\eg close to 10\% in F1).
Through our manual inspection, we find that sequence GINN's strength stands out when dealing with programs of greater complexity. Take programs in Figure~\ref{fig:metexs} as examples, even though both of them denote relatively simple semantics, sequence GNN struggles with their complexity (\eg large program size for Figure~\ref{fig:met1}\footnote{The average size of the AST in Java-small is around 100 while this program has 150+ nodes in its AST. 
}, and nested looping construct for Figure~\ref{fig:met2}). Therefore, its predictions do not precisely reflect the semantics of either program. sequence GNN predicts the name of the program in Figure~\ref{fig:met1} (\resp~\ref{fig:met2}) to be \texttt{SelectModifier} (\resp \texttt{CheckSoucre}) whereas sequence GINN produces correct predictions for both programs, which are highlighted within the shadow boxes.
\setlength{\intextsep}{15pt}
\begin{table}[h]
\centering
\caption{The first two rows show the results for sequence GINN and sequence GNN; next two rows are the results of the ablation study and the evaluation on alternative design.}
\vspace*{-2pt}
\adjustbox{max width=\columnwidth}{
\begin{tabular}{c|c|c|c|c}
\hline
Models & Configuration & F1 & ROGUE-2 & ROGUE-L \\\hline\hline
Sequence GNN & Original &51.3\% & 24.9\% & 49.8\%    \\\hline
\multirow{3}{*}{\textbf{Sequence GINN}} & Original    & \textbf{60.2\%} & \textbf{29.5\%} & \textbf{54.7\%} \\\cline{2-5}
    & Ablated     & 56.3\% & 26.3\% & 51.6\% \\\cline{2-5}
    & Alternative & 52.2\% & 24.8\% & 50.4\% \\\hline
\end{tabular}
}
\label{tab:resmet}
\end{table}

\vspace*{3pt}
\noindent
\textbf{\textit{Scalability.}}\,
In the same setup as the scalability experiment of variable misuse prediction task, we investigate the scalability of both models. To strengthen sequence GNN, we incorporate all additional edges to its program graphs~\cite{allamanis2017learning}. As depicted in Figure~\ref{fig:scamet}, under all three metrics, sequence GINN outperforms sequence GNN throughout the entire test set, especially in F1 where sequence GINN displays a wider margin over sequence GNN on larger graphs (\ie last 5 sets of the graphs except the seventh set) than it does on smaller graphs (\ie first 5 sets of graphs). 



\begin{figure*}[thb!]
\centering
  \begin{tikzpicture}[scale=0.72]
  \begin{axis}[
    xmax=10,xmin=1,
    ymin=0,ymax=1,
    xlabel=\emph{N-th Set (Ranked by Size of Graphs)},ylabel=\emph{F1},
    xtick={1,2,3,...,10},
    ytick={0,0.1,0.2,...,1},
    xticklabels={1,2,...,10},
  ]

     \addplot+ [mark=pentagon*,mark size=2.5pt,every mark/.append style={}] coordinates{(1, 0.65494) (2, 0.63971) (3, 0.58665) (4, 0.53307) (5, 0.49101) (6, 0.39495) (7, 0.47513) (8, 0.42398) (9, 0.42614) (10, 0.37709)};
     \addplot coordinates{(1, 0.70855) (2, 0.68790) (3, 0.65185) (4, 0.63421) (5, 0.58905) (6, 0.59201) (7, 0.52986) (8, 0.59218) (9, 0.57021) (10, 0.48576)};    
     \legend{Baseline, GINN}
    \end{axis}
    \end{tikzpicture}%
    ~%
    \begin{tikzpicture}[scale=0.72]
    \begin{axis}[
    xmax=10,xmin=1,
    ymin=0,ymax=1,
    xlabel=\emph{N-th Set (Ranked by Size of Graphs)},ylabel=\emph{ROGUE-2},
    xtick={1,2,...,10},
    ytick={0,0.1,...,1},
    xticklabels={1,2,...,10},
    ]
     \addplot+ [mark=pentagon*,mark size=2.5pt,every mark/.append style={}] coordinates{(1, 0.34) (2, 0.301) (3, 0.268) (4, 0.230034) (5, 0.20) (6, 0.22) (7, 0.24) (8, 0.23034) (9, 0.23) (10, 0.235)};
     \addplot coordinates{(1, 0.40115) (2, 0.38025) (3, 0.33455) (4, 0.27722) (5, 0.25000) (6, 0.24000) (7, 0.26103) (8, 0.26082) (9, 0.26769) (10, 0.25103)};
     \legend{Baseline, GINN}
    \end{axis}
    \end{tikzpicture}%
    ~%
     \begin{tikzpicture}[scale=0.72]
   \begin{axis}[
    xmax=10,xmin=1,
    ymin=0,ymax=0.9,
    xlabel=\emph{N-th Set (Ranked by Size of Graphs)},ylabel=\emph{ROGUE-L},
    xtick={1,2,...,10},
    ytick={0,0.1,...,0.9},
    xticklabels={1,2,...,10},
    ]

     \addplot+ [mark=pentagon*,mark size=2.5pt,every mark/.append style={}] coordinates{(1, 0.58125) (2, 0.57519) (3, 0.47026) (4, 0.46544) (5, 0.47480) (6, 0.42812) (7, 0.36565) (8,.36743) (9, 0.35743) (10, 0.35900)};
     \addplot coordinates{(1, 0.63238) (2, 0.61221) (3, 0.58926) (4, 0.57371) (5, 0.56970) (6, 0.48150) (7, 0.46928) (8,.44146) (9, 0.44990) (10, 0.48432)};
     \legend{Baseline, GINN}
     \end{axis} 
     \end{tikzpicture}
     \setlength{\belowcaptionskip}{-5pt}
     \setlength{\abovecaptionskip}{7pt}
    \caption{Investigating the scalability of all four neural architectures.}
    \label{fig:scamet}
\end{figure*}
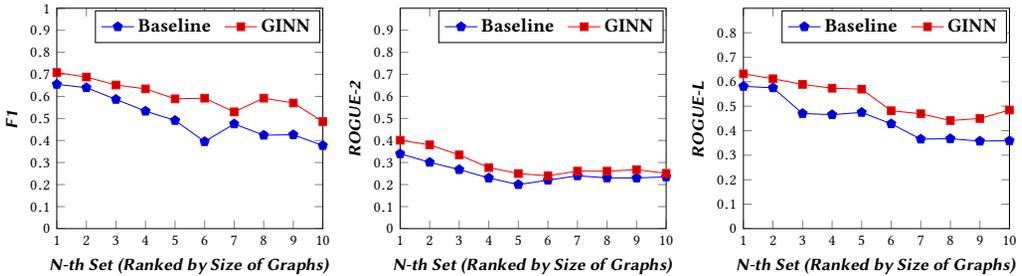

\vspace*{3pt}
\noindent
\textbf{\textit{Ablation Study.}}\,
We conduct the same ablation study to quantify the influence of the graph abstraction method on the performance of sequence GINN. As shown in Table~\ref{tab:resmet}, without the graph abstraction operators, sequence GINN still performs somewhat better than the baseline but notably worse than the original configuration, indicating the considerable influence of the abstraction method on sequence GINN's performance.

\vspace*{3pt}
\noindent
\textbf{\textit{Alternative Design.}}\,
Following the evaluation of the alternative design in variable misuse prediction, we adopt the same approach to initializing the state of freshly created nodes. The last row in Table~\ref{tab:resmet} shows that this design causes GINN to be less accurate by a fairly wide margin. 

\subsection{Detection of Null Pointer Dereference}
First, we present a framework for creating neural bug detectors in catching null pointer dereference. Then, we describe our experiment setup. Finally, we report our evaluation results. 

\vspace*{3pt}
\noindent
\textbf{\textit{Neural Bug Detection Framework.}}\, We devise a simpler program representation for this task. Given a control flow graph, we only split a node of basic block into multiple nodes of statement to aid the localization of bugs at the level of lines. When checking a method, we also include its context by stitching the graphs of all its callers and callees to that of itself.
Regarding the initial state of each statement node, we use RNN to learn a representation of its token sequence. As explained in Section~\ref{subsubsec:joint}, after each token is embedded into a numerical vector (by the initial embedding layer), we feed the embeddings of all tokens in a sequence to the network whose final hidden state (\ie the last $h_t$ in Equation~\ref{equ:rnn}) 
will be extracted as the initial node state. 

After GINN computes the state vector for each node, and in turn the entire graph (by aggregating the states of all nodes), we train a feed-forward neural network~\cite{svozil1997introduction} to predict (using the representation of the entire graph) whether or not a method is buggy. If it is, we will use the same network to predict the bug locations. Our rationale is a graph as a whole provides stronger signals for models to determine the correctness of a method. Correct methods will be refrained from further predictions, leading to fewer false warnings. On the other hand, if a method is indeed predicted to be buggy, our model will then pick the top $N$ statements ranked by their probabilities as potential bug locations. This prediction mode intends to provide flexibility for developers to tune their analysis toward either producing less false warnings or identifying more bugs. 

\vspace*{3pt}
\noindent
\textbf{\textit{Dataset.}}\, We target Java code due to the popularity and availability of Java datasets. In addition to bugs provided by existing datasets~\cite{just2014defects4j,saha2018bugs,ye2014learning}, we extract additional bugs from Bugswarm projects~\cite{tomassi2019bugswarm} to enrich our dataset for this experiment. 
For example, given a bug description "\textit{Issue \#2600...}" contained in a commit message, we search in the bug tracking system (\eg Bugzilla) using id \textit{\#2600} to retrieve the details of the bug, against which we match pre-defined keywords to determine if it's a null pointer dereference bug.
Next, we refer back to the commit to find out bug locations, specifically, we consider the lines that are modified by this commit to be buggy and the rest to be clean.
We acknowledge the location of a bug and its fix may not be precisely the same thing. If there are ever cases where two locations significantly differ, GINN and the compared baseline will be equally affected. 

Any detection is in nature a classification problem having subjects of interest one class of the data, and the rest the other. Therefore, in the context of bug detection, it's also necessary to supply the correct methods apart from the buggy methods for training. A natural approach would be directly taking the fixed version of each buggy method. However, this approach is unlikely to work well in practice. Because in a real problem setting, an analyzer will never set out to differentiate the two versions (\ie correct and buggy) of the same program, instead, it has to separate buggy programs from correct programs that are almost guaranteed to be functionally different. For this reason, we pair each buggy method with a correct method that is syntactically the closest (defined by the tree edit distance between their AST) from the same project. However, due to the shortage of bug instances in existing datasets even after taking into account the additional bugs we extracted ourselves, we include more correct methods \wrt each buggy method to further enlarge our dataset. Table~\ref{tab:stas} in Appendix~\ref{app:npedata} shows the details including the projects from which we extract the code snippet, a brief description of their functionality, size of their codebase, and the number of buggy methods we extract from each project.
We maintain an approximately 3:1 ratio between the correct and buggy methods for each project. In total, we use 64 projects for training and 13 for test. 



\vspace{3pt}
\noindent
\textbf{\textit{Objective, Metric, and Baseline.}}\,
We evaluate how accurately the neural bug detector can localize a null pointer dereference within a method. Because warning developers about buggy functions hardly makes a valuable tool.
As we deal with an unbalanced dataset (\ie cleaning lines are more than buggy lines), we forgo the accuracy metric since a bug detector predicting all lines to be correct would yield a decent accuracy but make a totally useless tool. Instead, we opt for precision, recall, and F1 score (Equation~\ref{equ:pre}--\ref{equ:f1} in Appendix~\ref{app:metrics}), metrics are commonly used in defect prediction literature~\cite{wang2016automatically,pradel2018deepbugs}. To show our interval-based graph abstraction method can improve a variety of graph models, we use the standard GNN to build a baseline bug detector for this experiment. We do not include classical static analysis tools as additional baselines
since comparing them against learning-based bug detectors 
is likely to be unfair due to the noise issue raised earlier. That is static analyzers could very well discover hidden bugs or report different locations of the same bug. Properly handling those situations require manual inspection, which is hard to scale. On the other hand, machine learning models are less susceptible to this problem as the training set yield in principle the same distribution of the test set. 

\vspace*{3pt}
\noindent
\textbf{\textit{Implementation.}}\, 
We construct the control flow and interval graphs using Spoon~\cite{pawlak:hal-01169705}, an open-source library for analyzing Java source code. To efficiently choose correct methods to pair with a buggy method, we run DECKARD~\cite{4222572}, a clone detection tool adopting an approximation of the shortest tree edit distance algorithm to identify similar code snippets. To realize the neural bug detector powered by GNN, we make two major changes to~\citet{allamanis2017learning}'s implementation of GGNN. First, we adopt the method defined in Equation~\ref{equ:type} and~\ref{equ:sum} for updating the states of nodes in the message-passing procedure.
Second,
we design two loss functions both in the form of cross-entropy for the bug detection task. The first one is designed for the classification of buggy and correct methods and the other for the classification of buggy and clean lines inside of a buggy method. Next, we realize the graph abstraction method within the implementation of GNN-powered bug detector to build GINN-powered bug detector. Similar to prior tasks, we implement only one abstraction cycle for GINN. For the remaining GGNN's parameters after our re-implementation, we keep them in both bug detectors.

\vspace*{3pt}
\noindent
\textbf{\textit{Performance.}}\, Table~\ref{tab:pre}--\ref{tab:f1} depict the precision, recall, and F1 for both bug detectors. We also examine the impact of the program context --- denoted by the number in parenthesis --- on the performance of each bug detector. More concretely, (0) means each method will be handled by itself. (1) stitches control flow graphs of all callers and callees (denoted by $\mathcal{F}$) to that of the target method, and (2) further integrates the graphs of all callers and callees of every method in $\mathcal{F}$. 
Exceeding 2 has costly consequences. First, most data points will have to be filtered out for overflowing the GPU memory. Moreover, a significant portion of the remaining ones would also have to be placed in a batch on its own, resulted in a dramatic increase in training time. We provide the same amount of context for each method in the test set as we do in the training set to maintain a consistent distribution across the entire dataset. Columns in each table correspond to three prediction modes in which models pick 1, 3, or 5 statements to be buggy after a method is predicted buggy.

\makebox[\dimexpr\linewidth-\parindent][s]{Overall, the neural bug detector built out of GINN consistently outperforms the baseline using}\par

\noindent
all proposed metrics, especially in F1 score where GINN beats GNN by more than 10\%. Regarding the impact of the context, we find that more context mostly but not always leads to improved
\begin{table*}[t]
\begin{minipage}{.32\textwidth}
\captionsetup{skip=1pt}
\caption{Precision.}
\centering
\adjustbox{max width=1\textwidth}{
\begin{tabular}{c|c|c|c}
\hline
 \tabincell{c}{Methods} & Top-1 & Top-3 & Top-5  \\\hline
GNN (0) & 0.209 & 0.147 & 0.117 \\\hline
GNN (1) & \textit{0.285} & 0.166 & 0.131 \\\hline
GNN (2) & 0.218 & 0.138 & 0.101 \\\hline
GINN (0) & 0.326 & 0.174 & 0.138 \\\hline
GINN (1) & \textbf{0.351} & 0.167 & 0.130 \\\hline
GINN (2) & 0.305 & 0.195 & 0.136 \\\hline
\hline
Gain & \multicolumn{3}{|l}{\hspace{1pt}\textbf{0.351} - \textit{0.285} = 0.066} \\\hline
\end{tabular}}
\label{tab:pre}
\end{minipage}
\begin{minipage}{.32\textwidth}
\captionsetup{skip=1pt}
\caption{Recall.}
\centering
\adjustbox{max width=1\textwidth}{
\begin{tabular}{c|c|c|c}
\hline
 \tabincell{c}{Methods} & Top-1 & Top-3 & Top-5  \\\hline
GNN (0) & 0.243 & 0.373 & \textit{0.450} \\\hline
GNN (1) & 0.171 & 0.252 & 0.302 \\\hline
GNN (2) & 0.210 & 0.337 & 0.375 \\\hline
GINN (0) & 0.282 & 0.374 & 0.447 \\\hline
GINN (1) & 0.329 & 0.431 & \textbf{0.507} \\\hline
GINN (2) & 0.304 & 0.428 & 0.500 \\\hline
\hline
Gain & \multicolumn{3}{|l}{\hspace{1pt}\textbf{0.507} - \textit{0.450} = 0.057} \\\hline
\end{tabular}}
\label{tab:re}
\end{minipage}
\begin{minipage}{.32\textwidth}
\captionsetup{skip=1pt}
\caption{F1 score.}
\centering
\adjustbox{max width=1\textwidth}{
\begin{tabular}{c|c|c|c}
\hline
 \tabincell{c}{Methods} & Top-1 & Top-3 & Top-5  \\\hline
GNN (0) & \textit{0.224} & 0.211 & 0.186 \\\hline
GNN (1) & 0.214 & 0.201 & 0.183 \\\hline
GNN (2) & 0.214 & 0.196 & 0.159 \\\hline
GINN (0) & 0.303 & 0.238 & 0.211 \\\hline
GINN (1) & \textbf{0.339} & 0.241 & 0.207 \\\hline
GINN (2) & 0.304 & 0.268 & 0.214 \\\hline
\hline
Gain & \multicolumn{3}{|l}{\hspace{1pt}\textbf{0.339} - \textit{0.224} = 0.115} \\\hline
\end{tabular}}
\label{tab:f1}
\end{minipage}
\end{table*}
\setlength{\textfloatsep}{12pt}

\noindent
performance of either bug detector. The reason is, on one hand, more information will always be beneficial. On the other hand, exceedingly large graphs hinders the generalization of graph models, resulted in the degraded performance of the neural bug detectors. 

\vspace*{3pt}
\noindent
\textbf{\textit{Scalability.}}\,
Similar to the prior studies, we investigate the scalability of both neural bug detectors. Again, we have strengthened the baseline by adding extra edges to its program graphs~\cite{allamanis2017learning}. We fix the context for each method to be 2, which provides the largest variation in graph size. Figure~\ref{fig:scaNPE} depicts the precision, recall, and F1 score for GINN and the baseline under top-1 prediction mode (top-3 and top-5 yield similar results). Even though GINN was beaten by GNN on few subsets of small graphs, it consistently outperforms the baseline as the size of graphs increases. Overall, we conclude bug detectors built out of GINN are more scalable.

\begin{figure*}[h]
\centering
  \begin{tikzpicture}[scale=0.72]
  \begin{axis}[
    xmax=10,xmin=1,
    ymin=0,ymax=0.7,
    xlabel=\emph{N-th Set (Ranked by Size of Graphs)},ylabel=\emph{Top-1 Precision},
    xtick={1,2,3,...,10},
    ytick={0,0.1,0.20,...,0.7},
    xticklabels={1,2,...,10},
  ]
     \addplot+ [mark=pentagon*,mark size=2.5pt,every mark/.append style={}] coordinates{(1, 0.513+0.08) (2, 0.375+0.08) (3, 0.224+0.08) (4, 0.1925+0.08) (5, 0.194+0.08) (6, 0.103+0.08) (7, 0.102+0.08) (8, 0.078+0.08) (9, 0.039+0.08) (10, 0.014+0.08)};
     \addplot coordinates{(1, 0.352+0.11) (2, 0.217+0.11) (3, 0.281+0.11) (4, 0.234+0.11) (5, 0.227+0.11) (6, 0.189+0.11) (7, 0.202+0.11) (8, 0.171+0.11) (9, 0.091+0.11) (10, 0.040+0.11)};   
     \legend{\emph{GNN}, \emph{GINN}}
    \end{axis}
    \end{tikzpicture}%
    ~%
    \begin{tikzpicture}[scale=0.72]
  \begin{axis}[
    xmax=10,xmin=1,
    ymin=0,ymax=0.7,
    xlabel=\emph{N-th Set (Ranked by Size of Graphs)},ylabel=\emph{Top-1 Recall},
    xtick={1,2,...,10},
    ytick={0,0.1,0.20,...,0.7},
    xticklabels={1,2,...,10},
    ]
     \addplot+ [mark=pentagon*,mark size=2.5pt,every mark/.append style={}] coordinates{(1, 0.453-0.12) (2, 0.396-0.12) (3, 0.402-0.12) (4, 0.383-0.12) (5, 0.377-0.12) (6, 0.388-0.12) (7, 0.330-0.12) (8, 0.338-0.12) (9, 0.329-0.12) (10, 0.295-0.12)};
     \addplot coordinates{(1, 0.432-0.12) (2, 0.458-0.12) (3, 0.464-0.12) (4, 0.455-0.12) (5, 0.454-0.12) (6, 0.448-0.12) (7, 0.429-0.12) (8, 0.426-0.12) (9, 0.404-0.12) (10, 0.413-0.12)};    
    \legend{\emph{GNN},\emph{GINN}}
    \end{axis}
    \end{tikzpicture}%
    ~%
     \begin{tikzpicture}[scale=0.72]
  \begin{axis}[
    xmax=10,xmin=1,
    ymin=0,ymax=0.7,
    xlabel=\emph{N-th Set (Ranked by Size of Graphs)},ylabel=\emph{Top-1 F1 Score},
    xtick={1,2,...,10},
    ytick={0,0.1,0.20,...,0.70},
    xticklabels={1,2,...,10},
    ]
    
     
     \addplot+ [mark=pentagon*,mark size=2.5pt,every mark/.append style={}] coordinates{(1, 0.426) (2, 0.344) (3, 0.292) (4, 0.268) (5, 0.265) (6, 0.217) (7, 0.195) (8, 0.183) (9, 0.152) (10, 0.122)};     
     \addplot coordinates{(1, 0.372) (2, 0.332) (3, 0.366) (4, 0.339) (5, 0.335) (6, 0.313) (7, 0.310) (8, 0.293) (9, 0.235) (10, 0.199)};
     
     \legend{\emph{GNN},\emph{GINN}}
     \end{axis} 
     \end{tikzpicture}
     \setlength{\abovecaptionskip}{6pt}
     \setlength{\belowcaptionskip}{-2pt}
    \caption{Investigating the scalability of baseline and GINN.}
    \label{fig:scaNPE}
\end{figure*}
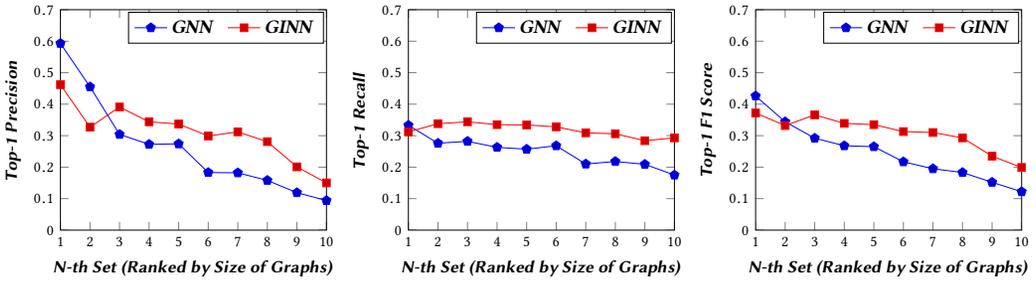

\vspace*{-1pt}
\noindent
\textbf{\textit{Catching Bugs in the Wild.}}\, We conduct another study to test if the GINN-based bug detector can discover new bugs in the real world. We set up both GINN-based bug detector (previously trained on 64 projects and configured with top-1 prediction mode) and Facebook Infer~\cite{calcagno2015moving,berdine2005smallfoot}, arguably the state-of-the-art static analysis tool for Java code, to scan the codebase of 20 projects on GitHub that are highly starred and actively maintained. Note that although Infer is sound by design (based on abstract interpretation~\cite{CousotCousot77-1}), it makes engineering compromises for the sake of usability, thus triggering both false positives and false negatives in practice. On the other hand, GINN-powered bug detector is by nature an unsound
tool, therefore we only study the utility of both tools in engineering terms. 

After manually inspecting all alarms produced by both tools, We confirmed 38 bugs out of 102 alarms GINN-based bug detector raises, which rounds up to a precision of 37.3\%. In comparison, Facebook Infer emitted 129 alarms out of which 34 are confirmed (\ie precision = 26.3\% ). To understand why our bug detector is more precise,
we make several important observations about

\noindent
the behavior of both tools through our manual inspection. All examples are included in Appendix~\ref{app:bugs}. 

First, Infer in general raised more alarms that are provably false (\eg Figure~\ref{fig:false1c},~\ref{fig:false2c}, and~\ref{fig:false3c}) and missed more bugs that are demonstrably real (\eg Figure~\ref{fig:true1} and~\ref{fig:true2}, both of which have been fixed after our reporting) than GINN-based bug detector. Second, Infer raised many alarms due to its confusion of some common APIs (\eg Figure~\ref{fig:API1c} and~\ref{fig:API2c}). Third, regarding path-sensitivity, there's still much room for infer to improve. For example, the path along which the bug is reported by Infer (Figure~\ref{fig:interc}) is clearly unsatisfiable. Specifically, the branch (Line~\ref{lines:if2}) that triggers the null pointer exception is guarded against the same condition under which the supposed null pointer is in fact instantiated properly (Line~\ref{lines:if1}). We have manually verified that no code in between changes the evaluation of the path condition at Line~\ref{lines:ifc2}. 
Unlike Facebook Infer, GINN-based bug detector searches for bug patterns it learned from the training data. Because a large number of the false alarms Infer reports clearly deviate 
from the norm exhibited in the training data, GINN-based bug detector is better at suppressing them, resulted in a higher precision. 

We also look into wrong predictions our bug detector made, and find that it tends to struggle when a bug pattern is not local --- the dereference is quite distant from where the null pointer is assigned. We conjecture that the cause is the signal networks receive gets weaker when a scattered bug pattern is overwhelmed by irrelevant statements. To address this issue, we apply the slicing technique~\cite{10.5555/800078.802557} and obtain encouraging results in a preliminary study. By simply slicing the graph of each tested method, the same model (without retraining) can now detect 39 bugs out of 93 alarms, a moderate improvement over the previous result. We plan to test this idea out more systematically and at a larger scale.



We reported the confirmed bugs (\ie 38 in total mentioned above) that GINN-based bug detector found, out of which 11 are fixed and another 12 are confirmed (fixes pending). Reasons for unconfirmed bugs are: (1) no response from developers; (2) requiring actual inputs that trigger the reported bugs; or (3) third party code developers did not write themselves. Our findings show that GINN-based bug detector is capable of catching new bugs in real-world Java projects.

\subsection{Discussion}
We summarize the lessons learned from our extensive evaluation. First and foremost, the interval-based graph abstraction method has shown to improve the generalization of various graph-based models (\eg standard GNN, GGNN, RNN Sandwich) across all three downstream PL tasks. In each task, the GINN-based model outperforms a GNN-based model --- on a dataset that the GNN-based model achieves state-of-the-art results --- by around 10\% in a metric adopted by the GNN-based model. Based on the collective experience of the community (Appendix~\ref{app:pub}), it is evident GINN has significantly advanced the state-of-the-art GNN as a general, powerful model in learning semantic program embeddings.
Second, our evaluation reveals the scalability issues existing graph models suffer from --- manifested in the significant performance drop against the increasing size of graphs --- and suggests an effective solution built out of the graph abstraction method. In fact, the overall accuracy improvement mentioned earlier is in large part attributed to GINN's efficient handling of the large graphs. Finally, our evaluation also shows model-based static analyzers can be a promising alternative to the classical static analysis tools for catching null pointer dereference bugs. Even though our results are still preliminary for drawing a general conclusion on the comparison between statistical-based and logic-based tools, it shows our neural bug detector is not only notably more precise than Infer but also capable of catching new bugs on many popular codebases on GitHub.

\section{Related Work}
\label{sec:rel}
In this section, we survey prior works on learning models of source code.~\citet{Hindle10} pioneer the field of machine learning for source code. In particular, Hindle~\etal find that programs that people write are mostly simple and rather repetitive, and thus they are amenable to statistical language models. Their finding is based on the evidence that the n-gram language model captures regularities in software as well as it does in English.


Later, many works propose to elevate the learning from the level of token sequences~\cite{Hindle10,Pu2016,AAAI1714603,Nguyen1145} to that of abstract syntax trees~\cite{maddison2014structured,Alon:2019:CLD:3302515.3290353,alon2018code2seq} in an attempt to capture the structure properties exhibited in source code. Notably,~\citet{Alon:2019:CLD:3302515.3290353} present a method for function name prediction. Specifically, it decomposes a program to a collection of paths in its abstract syntax tree, and learns the atomic representation of each path simultaneously with learning how to aggregate a set of them.

Nowadays, graph neural networks have become undoubtedly the most popular deep model of source code. Since the introduction of graph neural networks to the programming domain~\cite{li2015gated,allamanis2017learning}, they have been applied to a variety of PL tasks, such as program summarization~\cite{fernandes2018structured}, bug localization and fixing~\cite{Dinella2020HOPPITY}, and type inference~\cite{Wei2020LambdaNet}. Apart from being thoroughly studied and constantly improved in the machine learning field, GNN's capability of encoding semantic structure of code through a graph representation is a primary contributor to their success. While our work also builds upon GNN, we offer a fundamentally different perspective which no prior works have explored. That is program abstraction helps machine learning models to capture the essence of the semantics programs denote, and in turn facilitating the execution of downstream tasks in a precise and efficient manner. 

In parallel to all the aforementioned works, a separate line of works has emerged recently that use program executions (\ie dynamic models)~\cite{wang2017dynamic,wang2019learning,Wang101145} rather than source code (\ie static models) for learning program representations. Their argument is that source code alone may not be sufficient for models to capture the semantic program properties.~\citet{wang2019coset} show simple, natural transformations, albeit semantically-preserving, can heavily influence the predictions of models learned from source code. In contrast, executions that offer direct, precise, and canonicalized representations of the program behavior help models to generalize beyond syntactic features. On the flip side, dynamic models are likely to suffer from the low quality of training data since high-coverage executions are hard to obtain.~\citet{Wang101145} address this issue by blending both
source code and program executions.
Our work is set to improve static models via program abstraction, therefore we don't consider runtime information as a feature dimension, which can be an interesting future direction to explore.

\section{Conclusion}
\label{sec:con}

In this paper, we present a new methodology of learning models of source code. In particular, we argue by learning from abstractions of source code, models have an easier time to distill the key features for program representation, thus better serving the downstream tasks. At a technical level, we develop a principled interval-based abstraction method that directly applies to control flow graph. This graph abstraction method translates to a loop-based program abstraction at the source code level, which in essence makes models focus exclusively on looping construct for learning feature representations of source code. Through a comprehensive evaluation, we show our approach significantly outperforms the state-of-the-art models in two highly popular PL tasks: variable misuse and method name prediction. We also evaluate our approach in catching null pointer dereference bugs in Java programs. Results again show GINN-based bug detector not only beats the GNN-based bug detector but also yields a lower false positive ratio than Facebook Infer when deployed to catch null pointer deference bugs in the wild. We reported 38 bugs found by GINN to developers, among which 11 have been fixed and 12 have been confirmed (fixing pending). 

GINN is a general, powerful deep neural network that we believe is readily available to tackle a wide range of problems in program analysis, program comprehension, and developer productivity.

\bibliography{reference}
\clearpage
\appendix
   \begin{center}
      \huge\textbf{Appendix}
   \end{center}









\section{Loss Function}
\label{app:loss}

We present the loss function of the joint model below. First, we use the cross-entropy loss to specify the error between the predicted location and the true distribution, $Loc$, of the misuse variable.

\begin{equation*}
    H(\mathit{Loc},P_{\mathit{pro}}) = -\sum_{i\in [0,k)} \mathit{Loc}[i]\mathit{log}P_{\mathit{pro}}[i][0] = -\mathit{log}P_{\mathit{pro}}[i_{true}][0]
\end{equation*}
where $i_{true}$ is the actual index of the misuse variable in an input program. That is, the loss is the negative logarithm of $P_{\mathit{pro}}[i_{true}][0]$, the probability that the model assigns to the token at index $i_{true}$. As $P_{\mathit{pro}}[i_{true}][0]$ tends to 1, the loss approaches zero. The further $P_{\mathit{pro}}[i_{true}][0]$ goes below 1, the greater the loss becomes. Thus, minimizing this loss is equivalent to maximizing the log-likelihood that the model assigns to the true labels $i_{true}$. Similarly, the cross-entropy loss for repair variable is:
\begin{equation*}
    H(\mathit{Rep},P_{\mathit{pro}}) = -\sum_{i\in [0,k)} \mathit{Rep}[i]\mathit{log}P_{\mathit{pro}}[i][1] = -\mathit{log}P_{\mathit{pro}}[i_{true}][1]
\end{equation*}
The network will be trained to minimize both $H(\mathit{Loc},P_{\mathit{pro}})$ and $H(\mathit{Rep},P_{\mathit{pro}})$. For inference, the network predicts the token at $i_{\mathit{Loc}}$ (\resp $i_{\mathit{Rep}}$) to be the misuse (\resp repair) variable.
\begin{equation*}
\begin{split}
    & i_{\mathit{Loc}} = \argmax_{i\in [0,k)} P_{\mathit{pro}}[i][0]\\
    & i_{\mathit{Rep}} = \argmax_{i\in [0,k)} P_{\mathit{pro}}[i][1]
\end{split}
\end{equation*}

\clearpage

\section{Recurrent Neural Network}
\label{app:RNN}

A recurrent neural network (RNN)~\cite{rnnbook} is a class of artificial neural
networks that are distinguished from feedforward networks by their
feedback loops.  This allows RNNs to ingest their own outputs as
inputs. It is often said that RNNs have memory, enabling them to
process sequences of inputs.

Here we briefly describe the computation model of a vanilla RNN. Given
an input sequence, embedded into a sequence of vectors $x = (x_{1},
\cdot\cdot\cdot, x_{T_{x}})$, an RNN with $N$ inputs, a single hidden
layer with $M$ hidden units, and $Q$ output units. We define the RNN's
computation as follows:
\begin{align}
h_{t} &= f(W * x_{t} + V * h_{t-1}) \label{equ:1}\\
o_{t} &= \mathit{softmax}(Z * h_{t}) \notag
\end{align}
where $x_{t} \in \mathbb{R}^{N}$, $h_{t} \in \mathbb{R}^{M}$,
$o_{t}\in \mathbb{R}^{Q}$ is the RNN's input, hidden state and output
at time $t$, $f$ is a non-linear function (\eg tanh or sigmoid), $W
\in \mathbb{R}^{M*N}$ denotes the weight matrix for connections from
input layer to hidden layer, $V \in \mathbb{R}^{M*M}$ is the weight
matrix for the recursive connections (\ie from hidden state to itself)
and $Z \in \mathbb{R}^{Q*M}$ is the weight matrix from hidden to the
output layer.

\clearpage

\section{Partitioning the New Graphs into Intervals}
\label{app:parti}

We split our exposition into two steps: convert standard control flow graph to statement-based control flow graph and from statement-based control flow graph to token-based control flow graphs. Recall the definition of intervals: one node, called header, is the only entry node of a subgraph in which all closed paths contain the header node. In both steps, the intervals on the original control flow graph will be preserved because (1) there is no external node that connects to any statement or token node (2) the closed path is fundamentally the same as before except there will be more statement/token nodes on the path. Therefore, we say intervals consist of the same program statements regardless of the conversion from the standard to token-based control flow graph.

\clearpage

\section{Dataset Generation}
\label{app:dataset}

We used the ETH Py150 dataset~\cite{Raychev2984041}, which is based on GitHub Python code, and already partitioned into train and test splits (100K and 50K files, respectively). We further split the 100K train files into 90K train and 10K validation examples and applied a deduplication step on that dataset~\cite{Allamanis3359735}. We extracted all top-level function definitions from these files; any function that uses multiple variables can be turned into a training example by randomly replacing one variable usage with another. As there may be many candidates for such bugs in a function,~\citet{Hellendoorn2020Global} limited their extraction to up to three samples per function to avoid biasing the dataset too strongly towards longer functions. Since an important goal of our work is to improve the model scalability, we included more longer functions by creating 4 samples per function. To keep our dataset the same size as~\citet{Hellendoorn2020Global}'s (only for training and testing as they did not report the size of their validation set), we trimmed small functions from the other end. For every synthetically generated buggy example, an unperturbed, bug-free example of the function is included as well to keep our dataset balanced. Our dataset contains 2M programs for training, 245K for validation, and 755K for testing.

\clearpage

\section{Converting AST to GINN's graph representation}
\label{app:convert}

First, we walk each AST to distill all the statement nodes; then, for each statement node, we identify its immediate successor, such as the next statement within the same basic block, or the first statement of another basic block (due to control constructs like \texttt{if} statement or \texttt{for} loop).
Next, we add an edge from each statement node to its immediate successor. Finally, we replace statement nodes with sequences of token nodes, and move the
start and end node for each control flow edge from statement nodes to their first token nodes. 

Note that all programs used in~\cite{Hellendoorn2020Global} consist of single functions only, therefore we don't need to consider method calls when converting AST to CFG.

\clearpage

\section{Reproducing the Results Reported in Prior Works}
\label{app:rep}
We compare the performance of our own model implementations again those reported in~\cite{Hellendoorn2020Global}. In particular, we measure the independent localization and repair accuracy, the metrics adopted by~\cite{Hellendoorn2020Global}, of the core models instantiated by GGNN and RNN Sandwich. Table~\ref{tab:rep} shows the model performance is close across both neural architectures, signaling the validity of both our model implementations and data collection. 

Note that the accuracy below is what models achieved on the training set. Table~\ref{tab:res} in the paper shows the results of both GGNN and GGNN-powered RNN Sandwich on the test set, which is around 5\%--10\% lower than what's reported in~\cite{Hellendoorn2020Global}.\footnote{\citet{Hellendoorn2020Global} used classification accuracy and location+repair accuracy only for the test set.} We conjecture the accuracy drop is likely caused by the higher difficulty level of our testing set. As explained in Appendix~\ref{app:dataset}, our dataset contains more longer functions. Specifically, programs exceeding 250 tokens only make up 6.5\% of their test set~\cite{Hellendoorn2020Global} whereas programs of the same size take around 14\% in our test set. In addition, our experiment mainly focuses on the improvement of GINN over GGNN (or GNN). Since both models are implemented and executed with the exact same configurations (\ie server configuration, number of GPU cores, Tensorflow version, \etc), we believe our experimental results for variable misuse prediction task is legitimate.

\begin{table}[thb!]
\centering
\caption{}
\adjustbox{max width=\columnwidth}{
\begin{tabular}{c|c|c}
\hline
Models &  \tabincell{c}{Localization \\Accuracy} & \tabincell{c}{Repair \\Accuracy} \\\hline
GGNN (reported) & 79\% & 74\%    \\\hline
GGNN (self-implemented) & 77\% & 71\%    \\\hline 
\tabincell{c}{RNN Sandwich \\ (reported)} & 81\% & 86\% \\\hline
\tabincell{c}{RNN Sandwich \\ (self-implemented)} & 80\% & 84\% \\\hline
\end{tabular}
}
\label{tab:rep}
\end{table}

\clearpage

\section{Training Time}
\label{app:tt}

In this section, we report the training time for every evaluated model under each program analysis task. We ignore the inference time as all models make predictions instantaneously.

\vspace*{4pt}
\noindent
\textbf{\textit{Variable Misuse Prediction}}\, Figure~\ref{fig:ttvar} presents the training time of all evaluated models in variable misuse prediction task. We use the localization accuracy as other metrics show the same trend.

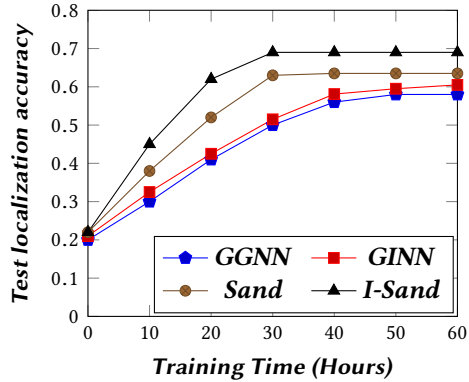
\begin{figure*}[h]
\centering
   \begin{tikzpicture}[scale=1]
   \begin{axis}[
    xmax=60,xmin=0,
       ymin=0,ymax=0.8,
    xlabel=\emph{Training Time (Hours)},ylabel=\emph{Test localization accuracy},
    xtick={0,10,20,...,60},
    ytick={0,0.1,0.2,0.3,0.4,...,0.8},
    xticklabel={\pgfmathparse{\tick}\pgfmathprintnumber{\pgfmathresult}},
    legend style={
        at={(0.99,0.01)},
        anchor=south east,
        font=\bfseries,
        /tikz/column 2/.style={
            column sep=5pt,
        },        
    },
   ]
   
     \addplot+ [mark=pentagon*,mark size=2.5pt,every mark/.append style={}] coordinates{(0, 0.3-.1)  (10, 0.5-0.2)  (20, 0.61-0.2) (30, 0.70-0.2)  (40, 0.76-0.2)  (50, 0.78-0.2)  (60, 0.78-0.2)};
     \addplot coordinates{(0, 0.35-.14)  (10, 0.51-.185)  (20, 0.59-.165) (30, 0.68-.165)  (40, 0.78-.199)  (50, 0.81-.215)  (60, 0.82-.215)};
     \addplot coordinates{(0, 0.32-.1)  (10, 0.58-0.2)  (20, 0.72-0.2) (30, 0.83-0.2)  (40, 0.835-0.2)  (50, 0.835-0.2)  (60, 0.835-0.2)};     
     \addplot+ [mark=triangle*,mark size=2.5pt] coordinates{(0, 0.32-.1)  (10, 0.65-.2)  (20, 0.82-.2) (30, 0.89-.2)  (40, 0.89-.2)  (50, 0.89-.2)  (60, 0.89-.2)};     

     \legend{\emph{GGNN}, \emph{GINN}, \emph{Sand}, \emph{I-Sand}}
    \end{axis}
    \end{tikzpicture}%
    \caption{Models' training time for variable misuse prediction.}
    \label{fig:ttvar}
\end{figure*}

\noindent
\textbf{\textit{Method Name Prediction}}\, Figure~\ref{fig:ttmn} presents the training time of all evaluated models in the method name prediction task.

\begin{figure*}[h]
\centering
   \begin{tikzpicture}[scale=1]
   \begin{axis}[
    xmax=120,xmin=0,
    ymin=0,ymax=0.8,
    xlabel=\emph{Training Time (Minutes)},ylabel=\emph{F1 achieved on test set},
    xtick={0,20,40,...,120},
    ytick={0,0.2,0.40,...,0.8},
    xticklabel={\pgfmathparse{\tick}\pgfmathprintnumber{\pgfmathresult}},
    ]
     \addplot+ [mark=pentagon*,mark size=2.5pt,every mark/.append style={}] coordinates{ (0, 0.193) (20, 0.31) (40, 0.4) (60, 0.45) (80, 0.49) (100, 0.51) (120, 0.51)};
     \addplot coordinates{ (0, 0.22) (20, 0.35) (40, 0.46) (60, 0.56) (80, 0.6) (100, 0.61) (120, 0.61)};   
    \legend{\emph{GGNN},\emph{GINN}}
    \end{axis}
    \end{tikzpicture}%
    \caption{Models' training time for method name prediction.}
    \label{fig:ttmn}
\end{figure*}
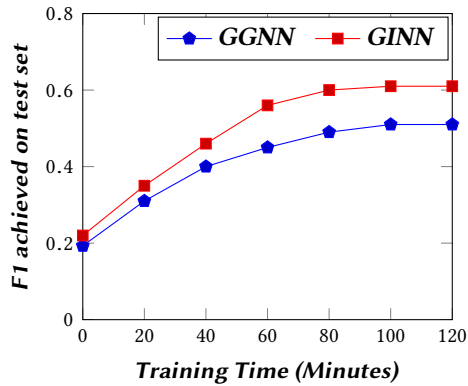

\noindent
\textbf{\textit{Null Pointer Dereference}}\, Figure~\ref{fig:ttnpe} presents the training time of all evaluated models in null pointer dereference detection task. Y-axis denotes precision score models achieved on the test set (with context (1)). Recall and F1 shows a similar trend. 

\begin{figure*}[h]
\centering
   \begin{tikzpicture}[scale=1]
   \begin{axis}[
    xmax=60,xmin=0,
    ymin=0,ymax=0.5,
    xlabel=\emph{Training Time (Minutes)},ylabel=\emph{Test Precision in Top-1},
    xtick={0,10,20,...,60},
    ytick={0,0.1,0.20,...,0.50},
    xticklabel={\pgfmathparse{\tick}\pgfmathprintnumber{\pgfmathresult}},
    ]
     \addplot+ [mark=pentagon*,mark size=2.5pt,every mark/.append style={}] coordinates{ (0, 0.119) (10, 0.145) (20, 0.191) (30, 0.223) (40, 0.255) (50, 0.278) (60, 0.279) (70, 0.278) (80, 0.278) (90, 0.319) (100, 0.319)};
     \addplot coordinates{(0, 0.123) (10, 0.20) (20, 0.265) (30, 0.315) (40, 0.351) (50, 0.350) (60, 0.350) (70, 0.409) (80, 0.412) (90, 0.410) (100, 0.411)};   
     \legend{\emph{GNN},\emph{GINN}}
     \end{axis} 
     \end{tikzpicture}%
    \caption{Models' training time for null pointer dereference detection.}
    \label{fig:ttnpe}
\end{figure*}
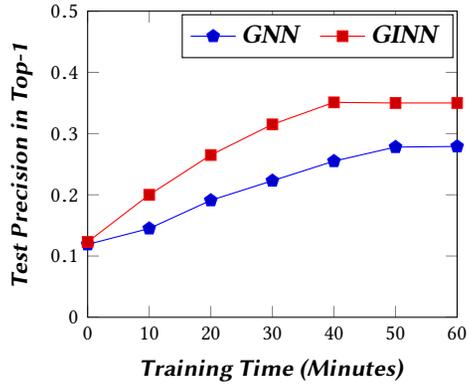

\clearpage

\section{Addition edges proposed by Allamanis \etal}
\label{app:edges}

Below we list the edges~\citet{allamanis2017learning} designed on top of AST.

\begin{itemize}
    \item \textbf{NextToken} connects each terminal node (syntax token) to its successor.
    \item \textbf{LastRead} connects a terminal node of a variable to all elements of the set of terminal nodes at which the variable could have been read last.
    \item \textbf{LastWrite}: connects a terminal node of a variable to all elements of the set of syntax tokens at which the variable was could have been last written to.
    \item \textbf{ComputedFrom} connects a terminal node of a variable $v$ to all variable tokens occurring in $expr$ when $expr$ is assigned to $v$. 
    \item \textbf{LastLexicalUse} chains all uses of the same variable.
    \item \textbf{ReturnsTo} connects \texttt{return} tokens to the method declaration.
    \item \textbf{FormalArgName} connects arguments in method calls to the formal parameters that they are matched to.
    \item \textbf{GuardedBy} connects every token corresponding to a variable (in the true branch of a \texttt{if} statement) to the enclosing guard expression that uses the variable.
    \item \textbf{GuardedByNegation} connects every token corresponding to a variable (in the false branch of a \texttt{if} statement) to the enclosing guard expression that uses the variable.    
\end{itemize}

\clearpage

\section{Metrics used in Method name prediction}
\label{app:mnmetrics}

We adopted the measure used by previous works~\cite{alon2018code2seq,Alon:2019:CLD:3302515.3290353,Wang101145,fernandes2018structured} in method name prediction, which measured F1 score (\ie the harmonic mean of precision and recall) over subtokens, case-insensitive. The intuition is the quality of a method name prediction largely depends on the constituent sub-words. For example, for a method called \texttt{computeDiff}, a prediction of \texttt{diffCompute} is considered as an exact match, a prediction of \texttt{compute} has a full precision but low recall, and a prediction of \texttt{computeFileDiff} has a full recall but low precision. 

\textit{Below is extracted from \url{https://en.wikipedia.org/wiki/ROUGE_(metric)}.}

ROUGE, which stands for Recall-Oriented Understudy for Gisting Evaluation, is a set of metrics for evaluating automatic summarization and machine translation in natural language processing. The metrics compare an automatically produced summary or translation against a reference or a set of references (human-produced) summary or translation.

\begin{itemize}
    \item ROUGE-1 refers to the overlap of unigram (each word) between the system and reference summaries.
    \item ROUGE-2 refers to the overlap of bigrams between the system and reference summaries.
    \item ROUGE-L: Longest Common Subsequence based statistics. Longest common subsequence problem takes into account sentence level structure similarity naturally and identifies longest co-occurring in sequence n-grams automatically.
\end{itemize}

\clearpage

\section{Dataset for NPE detection task}
\label{app:npedata}

We give the details about the dataset we assembled for the evaluation of detecting null pointer dereference.

\begin{table*}[h]
\caption{}
\adjustbox{max width=.8\paperwidth}{
\centering
\small
\begin{tabular}{c|c|c|c|c}
\hline
 Dataset & Projects & Description & Size (KLoC) & Number of Buggy Methods \\ \hline
\multirow{13}{*}{Test} &Lang & Java lang library &50 & 7   \\\cline{2-5}
&Closure & A JavaScript checker and optimizer &260 & 18  \\\cline{2-5}
&Chart & Java chart library &149 & 17  \\\cline{2-5}
&Mockito & Mocking framework for unit tests &45 & 14  \\\cline{2-5}
&Math & Mathematics and statistics components &165 & 16   \\\cline{2-5}
&Accumulo & Key/value store &194 & 6   \\\cline{2-5}
&Camel & Enterprise integration framework &560 & 17   \\\cline{2-5}
&Flink & System for data analytics in clusters &258 & 13   \\\cline{2-5}
&Jackrabbit-oak & hierarchical content repository &337 & 23  \\\cline{2-5}
&Log4j2 & Logging library for Java &70 & 29   \\\cline{2-5}
&Maven & Project management and comprehension tool & 62 & 6   \\\cline{2-5}
&Wicket & Web application framework &206 & 7   \\\cline{2-5}
&Birt & Data visualizations platform &1,093 & 678   \\\hline
\hline
\multicolumn{4}{r|}{Number of Buggy Methods in Total for Test}& 793  \\ 
\multicolumn{4}{r|}{Number of Methods in Total for Test}& 3,000      \\\hline 
\hline
\multirow{2}{*}{Validation} &SWT & Eclipse Platform project repository & 460 & 276 \\\cline{2-5}
&Tomcat & Web server and servlet container &222 & 111   \\\cline{2-5}
\hline\hline
\multicolumn{4}{r|}{Number of Buggy Methods in Total for Validation}& 387  \\ 
\multicolumn{4}{r|}{Number of Methods in Total for Validation}& 1,000      \\\hline 
\hline
\multirow{4}{*}{Training} &JDT UI & User interface for the Java IDE & 508 & 897  \\\cline{2-5}
&Platform UI & User interface and help components of Eclipse & 595 & 920  \\\cline{2-5}
&AspectJ & An aspect-oriented programming extension &289 & 151  \\\cline{2-5}
&from BugSwarm & 61 projects on GitHub & 5,191 & 156  \\\hline
\hline
\multicolumn{4}{r|}{Number of Buggy Methods in Total for Training}& 2,124   \\
\multicolumn{4}{r|}{Number of Methods in Total for Training}& 9,000   \\\hline
\end{tabular}
}
\label{tab:stas}
\end{table*}

\clearpage

\section{Precision, recall and F1 score}
\label{app:metrics}

Below we explain precision, recall and F1 score for bug detection.

\begin{align}
\mathit{Precision} &=  \frac{\mathit{TP}}{\mathit{TP} + \mathit{FP}} \label{equ:pre}\\ 
\mathit{Recall} &=  \frac{\mathit{TP}}{\mathit{TP} + \mathit{FN}} \\
\mathit{F1\,\, Score} &=  \frac{2 * \mathit{Precision} * \mathit{Recall}}{\mathit{Precision} + \mathit{Recall}} \label{equ:f1}
\end{align}

TP, FP, and FN denote true positive, false positive, and false negative respectively.
True positive is the number of the predicted buggy lines that are truly buggy, while false
positive is the number of predicted buggy line that are not buggy. False negative records
the number of buggy lines that are predicted as non-buggy. A higher precision suggests
a relatively low number of false alarms while high recall indicates a relatively low number of missed
bugs. F1 takes both precision and recall into consideration.

\clearpage

\section{Manual Investigation of the Bug reports}
\label{app:bugs}

Below is the program extracted from \url{https://github.com/crossoverJie/JCSprout/blob/master/src/main/java/com/crossoverjie/algorithm/TwoSum.java}.
\begin{figure}[h]
\begin{minipage}{0.9\textwidth}
  \begin{center}
	\lstset{style=mystyle}
\begin{lstlisting}[linewidth=6.5cm,numbers=left,firstnumber=48,framexleftmargin=7pt]
public int[] getTwo2(int[] nums,int target){
    int[] result = new int[2] ;
    Map<Integer,Integer> map = new HashMap<>(2) ;
    for (int i=0;i<nums.length;i++){

        if (map.containsKey(nums[i])){
            result = new int[]{map.get(nums[i]),i} ;
        }
        map.put(target - nums[i],i) ;
    }
    return result ;
}
\end{lstlisting}
  \end{center}
\end{minipage}
\caption{}
\label{fig:false1c}
\end{figure}

\noindent
Below is the report produced by Facebook Infer for the program in Figure~\ref{fig:false1c}. The alarm is proven to be false.
\begin{figure}[htb!]
\centering
  \includegraphics[width=.80\linewidth]{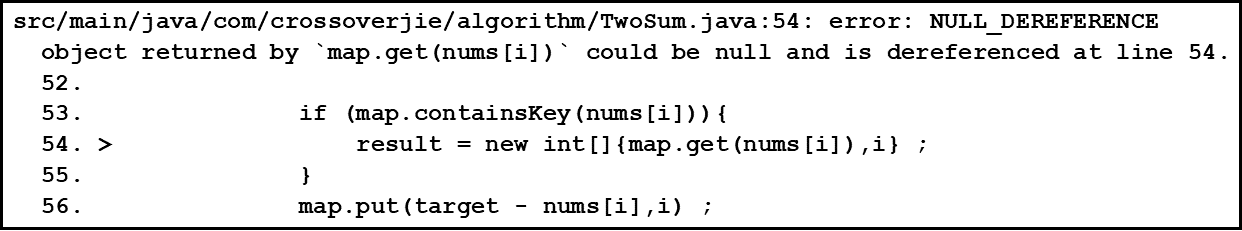}
  \caption{}
  \label{fig:false1}
\end{figure}

\noindent
Below is the program extracted from \url{https://github.com/alibaba/Sentinel/blob/master/sentinel-core/src/main/java/com/alibaba/csp/sentinel/node/metric/MetricTimerListener.java}.
\begin{figure}[h]
\begin{minipage}{0.9\textwidth}
	\begin{center}
	\lstset{style=mystyle}
\begin{lstlisting}[linewidth=13.6cm,numbers=left,firstnumber=59,framexleftmargin=7pt]
private void aggregate(Map<Long, List<MetricNode>> maps, Map<Long, MetricNode> metrics, ClusterNode node) {
    for (Entry<Long, MetricNode> entry : metrics.entrySet()) {
        long time = entry.getKey();
        MetricNode metricNode = entry.getValue();
        metricNode.setResource(node.getName());
        metricNode.setClassification(node.getResourceType());
        if (maps.get(time) == null) {
            maps.put(time, new ArrayList<MetricNode>());
        }
        List<MetricNode> nodes = maps.get(time);
        nodes.add(entry.getValue());
    }
}
\end{lstlisting}
  \end{center}
\end{minipage}
\caption{}
\label{fig:false2c}
\end{figure}

\noindent
Figure~\ref{fig:false2} shows the report produced by Facebook Infer for the program in Figure~\ref{fig:false2c}. This alarm is also proven to be false.
\begin{figure}[htb!]
\centering
  \includegraphics[width=.9\linewidth]{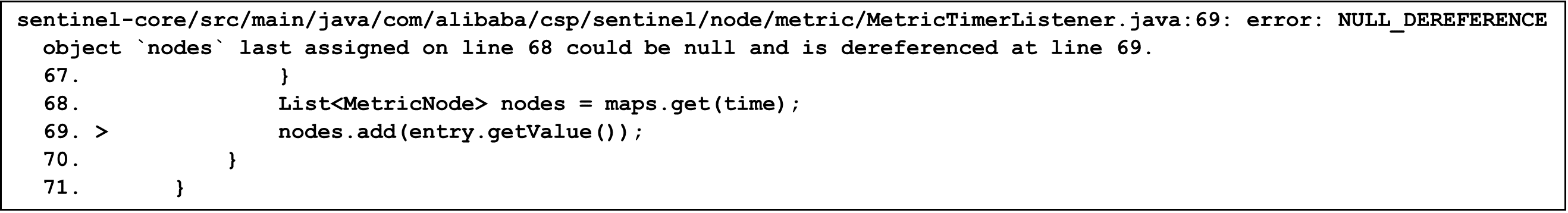}
  \caption{}
  \label{fig:false2}
\end{figure}

\noindent
Figure~\ref{fig:false3c} shows the program extracted from \url{https://github.com/stagemonitor/stagemonitor/blob/0.89.0/stagemonitor-tracing/src/main/java/org/stagemonitor/tracing/profiler/CallStackElement.java}.
\begin{figure}[h]
\begin{minipage}{0.9\textwidth}
	\begin{center}
	\lstset{style=mystyle}
\begin{lstlisting}[linewidth=11.2cm,numbers=left,firstnumber=75,framexleftmargin=9pt]
public void recycle() {
	if (!useObjectPooling) {
		return;	
	}	
	parent = null;
	@signature = null;@
	executionTime = 0;
	for (CallStackElement child : children) {
		child.recycle();	
	}
	children.clear();
	objectPool.offer(this);	
}
	
public void removeCallsFasterThan(long thresholdNs) {
	for (Iterator<CallStackElement> iterator = children.iterator(); iterator.hasNext(); ) {
		CallStackElement child = iterator.next();
		if (child.executionTime < thresholdNs && !@child.isIOQuery()@) {
			iterator.remove();
			@child.recycle();@
		} else {
			child.removeCallsFasterThan(thresholdNs);
		}
	}
}

public boolean isIOQuery() {
	// that might be a bit ugly, but it saves reference to a boolean and thus memory
	return @signature@.charAt(signature.length() - 1) == ' ';	
}
\end{lstlisting}
  \end{center}
\end{minipage}
\caption{}
\label{fig:false3c}
\end{figure}

\noindent
Figure~\ref{fig:false3} is the report produced by Facebook Infer for the program in Figure~\ref{fig:false3c}. The reason Infer think a npe is thrown by \texttt{child.isIOQuery} on line 92 is because inside of \texttt{child.recycle()}, called at line 94, \texttt{signature} will be assigned as null (line 80). Then in the next iteration when \texttt{child.isIOQuery} is invoked \texttt{charAt} function will be called on a null pointer at line 103. However, this report is provably false because the statement at line 91 guarantees \texttt{child} object refers to different memory cells in an array during each iteration, meaning \texttt{recyle} function will never be called on the same object as \texttt{isQuery}. 
\begin{figure}[htb!]
\centering
  \includegraphics[width=.80\linewidth]{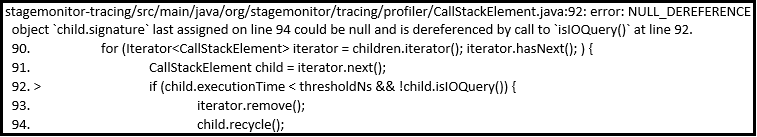}
  \caption{}
  \label{fig:false3}
\end{figure}

\noindent
Figure~\ref{fig:true1} shows the program extracted from \url{https://github.com/winder/Universal-G-Code-Sender/blob/54a1db413980d4c90b5b704c2723f0e740de20be/ugs-core/src/com/willwinder/universalgcodesender/i18n/Localization.java}. When \texttt{region} is null at line 7 inside of function \texttt{getString}, a null pointer exception will be thrown in function initialize at line 2. This bug, which has been fixed after our reporting, is missed by Facebook Infer but caught by GINN-based bug detector.
\begin{figure}[h]
\begin{minipage}{0.9\textwidth}
	\begin{center}
	\lstset{style=mystyle}
\begin{lstlisting}[linewidth=8.1cm,numbers=left,firstnumber=1,framexleftmargin=7pt]
synchronized public static boolean initialize(String language) {
    String[] lang = language.split("_");
    return initialize(lang[0], lang[1]);
}

public static String getString(String id, String region) {
    if (region == null || !region.equals(Localization.region)) {
        initialize(region);
    }
    return getString(id);
}
\end{lstlisting}
  \end{center}
\end{minipage}
\caption{}
\label{fig:true1}
\end{figure}

\noindent
Figure~\ref{fig:true2} is the program extracted from \url{https://github.com/brianfrankcooper/YCSB/blob/cd1589ce6f5abf96e17aa8ab80c78a4348fdf29a/jdbc/src/main/java/site/ycsb/db/JdbcDBClient.java}. If \texttt{driver} is checked against null at line 9, it is reasonable to add a check before calling \texttt{driver.contains} at line 3. Like the previous bug, this one is also missed by Facebook Infer but caught by GINN-based bug detector, in addition, it is fixed after we report the bug.

\begin{figure}[h]
\begin{minipage}{0.9\textwidth}
	\begin{center}
	\lstset{style=mystyle}
\begin{lstlisting}[linewidth=6.1cm,numbers=left,firstnumber=1,framexleftmargin=7pt]
String driver = props.getProperty(DRIVER_CLASS);

if (driver.contains("sqlserver")) {
  sqlserver = true;
}

...
try {
  if (driver != null) {
    Class.forName(driver);
  }

...
}

 
\end{lstlisting}
  \end{center}
\end{minipage}
\caption{}
\label{fig:true2}
\end{figure}

\noindent
Figure~\ref{fig:API1c} is the program extracted from \url{https://github.com/ctripcorp/apollo/blob/master/apollo-core/src/main/java/com/ctrip/framework/foundation/internals/Utils.java}. This alarm is provably false, and the reason seems that Facebook Infer is confused by the API \texttt{Strings.nullToEmpty} in \texttt{com.google.common.base.Strings}.
\begin{figure}[h]
\begin{minipage}{0.9\textwidth}
	\begin{center}
	\lstset{style=mystyle}
\begin{lstlisting}[linewidth=6.6cm,numbers=left,firstnumber=5,framexleftmargin=7pt]
public class Utils {
	public static boolean isBlank(String str) {
		return Strings.nullToEmpty(str).trim().isEmpty();
	}

	public static boolean isOSWindows() {
		String osName = System.getProperty("os.name");
		if (Utils.isBlank(osName)) {
			return false;
		}
		return osName.startsWith("Windows");
	}
}
\end{lstlisting}
  \end{center}
\end{minipage}
\caption{}
\label{fig:API1c}
\end{figure}

\noindent
Figure~\ref{fig:API1} shows the report produced by Facebook Infer for the program in Figure~\ref{fig:API1c}. This is a false alarm.
\begin{figure}[htb!]
\centering
  \includegraphics[width=.9\linewidth]{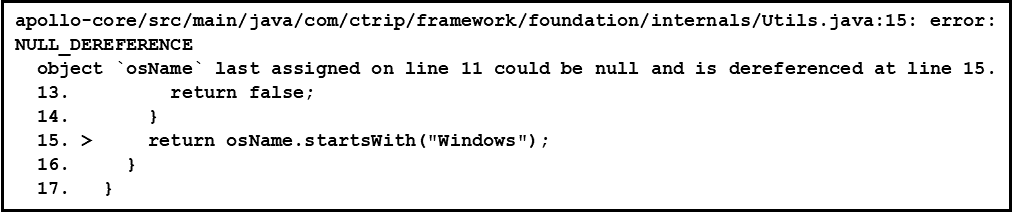}
  \caption{}
  \label{fig:API1}
\end{figure} 

\noindent
Figure~\ref{fig:API2c} is the program extracted from \url{https://github.com/brianfrankcooper/YCSB/blob/master/s3/src/main/java/site/ycsb/db/S3Client.java}. This alarm is provably false, and the reason seems that Facebook Infer is confused by the API \texttt{propsCL.getProperty("s3.protocol", "HTTPS");} (a method in \texttt{java.util.Properties}) at line 210, which assigns a default value \texttt{HTTPS} to \texttt{protocol}. 
\begin{figure}[h]
\begin{minipage}{0.9\textwidth}
	\begin{center}
	\lstset{style=mystyle}
\begin{lstlisting}[linewidth=7.9cm,numbers=left,firstnumber=208,framexleftmargin=10pt]
protocol = props.getProperty("s3.protocol");
if (protocol == null){
    protocol = propsCL.getProperty("s3.protocol", "HTTPS");    
}

...
clientConfig = new ClientConfiguration();
clientConfig.setMaxErrorRetry(Integer.parseInt(maxErrorRetry));
if(protocol.equals("HTTP")) {
	clientConfig.setProtocol(Protocol.HTTP); 
} else {
	clientConfig.setProtocol(Protocol.HTTPS);  
}
\end{lstlisting}
  \end{center}
\end{minipage}
\caption{}
\label{fig:API2c}
\end{figure}

\noindent
Figure~\ref{fig:API2} shows the report produced by Facebook Infer for the program in Figure~\ref{fig:API2c}. This is another false alarm.
\begin{figure}[htb!]
\centering
  \includegraphics[width=.9\linewidth]{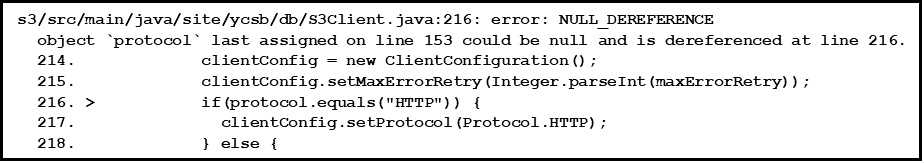}
  \caption{}
  \label{fig:API2}
\end{figure}

\noindent
Figure~\ref{fig:interc} is the program extracted from \url{https://github.com/stagemonitor/stagemonitor/blob/0.89.0/stagemonitor-web-servlet/src/main/java/org/stagemonitor/web/servlet/filter/HttpRequestMonitorFilter.java}. The path on which Infer reports a NPE is infeasible. \texttt{isInjectContentToHtml(request)} if evaluated to be \texttt{true}, \texttt{httpServletResponseBufferWrapper} was instantiated properly under the same condition. We have manually verified that no code in between the two \texttt{isInjectContentToHtml(request)} (line 99 and 111) affects the output of \texttt{isInjectContentToHtml(request)}. We have attached a detailed call chain for \texttt{monitorRequest} in the supplemental material. The remaining part of the code is straight-forward, and can be easily checked.

\begin{figure}[h]
\begin{minipage}{0.9\textwidth}
	\begin{center}
	\lstset{style=mystyle}
\begin{lstlisting}[linewidth=13.5cm,numbers=left,firstnumber=95,framexleftmargin=10pt]
private void doMonitor(HttpServletRequest request, HttpServletResponse response, FilterChain filterChain) {

	final StatusExposingByteCountingServletResponse responseWrapper;
	HttpServletResponseBufferWrapper httpServletResponseBufferWrapper = null;
	if (isInjectContentToHtml(request)) {
		httpServletResponseBufferWrapper = new HttpServletResponseBufferWrapper(response); %\label{lines:if1}%
		responseWrapper = new StatusExposingByteCountingServletResponse(httpServletResponseBufferWrapper);
	} else {
		responseWrapper = new StatusExposingByteCountingServletResponse(response);		
	}

	try {
		monitorRequest(filterChain, request, responseWrapper);
	} catch (Exception e) {
		handleException(e);
	} finally {
		if (isInjectContentToHtml(request)) { %\label{lines:ifc2}%
			injectHtml(response, request, httpServletResponseBufferWrapper);  %\label{lines:if2}%
		}	
	}	
}
\end{lstlisting}
  \end{center}
\end{minipage}
\caption{}
\label{fig:interc}
\end{figure}

\noindent
Figure~\ref{fig:inter} shows the report produced by Facebook Infer for the program in Figure~\ref{fig:interc}. This is a false alarm.
\begin{figure}[htb!]
\centering
  \includegraphics[width=.8\linewidth]{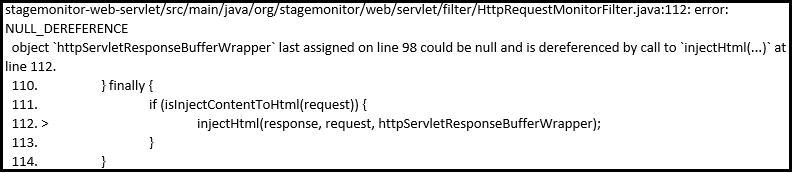}
  \caption{}
  \label{fig:inter}
\end{figure}

\noindent
List of fixed and confirmed bugs:

\noindent
\url{https://github.com/ctripcorp/apollo/issues/3012} \\
\url{https://github.com/spring-projects/spring-boot/issues/20890}\\
\url{https://github.com/SynBioDex/libSBOLj/issues/601}\\
\url{https://github.com/SynBioDex/libSBOLj/issues/603}\\
\url{https://github.com/konsoletyper/teavm/issues/438}\\
\url{https://github.com/brianfrankcooper/YCSB/issues/1371}\\
\url{https://github.com/ome/bioformats/issues/3464}\\
\url{https://github.com/nutzam/nutz/issues/1532}\\
\url{https://github.com/winder/Universal-G-Code-Sender/issues/1304}\\
\url{https://github.com/SeleniumHQ/selenium/issues/8183}\\
\url{https://github.com/SeleniumHQ/selenium/issues/8184}\\
\url{https://github.com/SeleniumHQ/selenium/issues/8185}\\
\url{https://github.com/kdn251/interviews/pull/159}\\
\url{https://github.com/spring-projects/spring-boot/issues/20901}\\
\url{https://github.com/crossoverJie/JCSprout/issues/195}\\

\clearpage

\section{publication history}
\label{app:pub}

\textit{\textbf{Variable Misuse Prediction}} GGNN~\cite{allamanis2017learning} is the first notable model, which achieves 74.1\% accuracy on a subset of MSR-VarMisuse. As depicted in Table~\ref{tab:111}, on the exact same dataset using~\citet{allamanis2017learning}'s metric, the joint model proposed by~\citet{vasic2018neural} is in fact notably worse than GGNN. The joint model performed better than GGNN on a new dataset with a set of new metrics.

\begin{table*}[h]
\caption{}
\adjustbox{max width=.8\paperwidth}{
\centering
\small
\begin{tabular}{lcccc}
\hline
\multirow{2}{*}{Model}   & MSR-VarMisuse  & \multicolumn{3}{c}{ETH-Py150} \\\cline{2-5}
                         & \tabincell{c}{end-to-end \\Accuracy} & \tabincell{c}{Classification \\Accuracy} & \tabincell{c}{Localization \\Accuracy} & \tabincell{c}{Localization+Repair \\Accuracy} \\\hline
GGNN                     & 71.4\%         & 71.1\%   & 64.6\%   & 55.8\% \\
Joint Model              & 62.3\%         & 82.4\%   & 71\%     & 65.7\% \\\hline\hline
\textbf{Improvement}     & \textbf{-9.1\%}& \textbf{11.3\%}   & \textbf{6.4\%}    & \textbf{9.9\%}  \\\hline
\end{tabular}
}
\label{tab:111}
\end{table*}

Given the conceptual framework laid out by the joint model,~\citet{Hellendoorn2020Global} proposed an engineering upgrade by simply replacing the RNN with a Transformer~\cite{vaswani2017attention} as the new instantiation of the core model. Table~\ref{tab:222} shows how RNN Sandwich (the state-of-the-art) compared to the engineering upgrade of the model proposed by~\citet{vasic2018neural}. Note that~\citet{Hellendoorn2020Global} curated a different dataset than~\citet{vasic2018neural}'s from ETH-Py150. $\leq$250 (\resp 1000) represents the set of programs that consist of less than 250 (\resp 1000) tokens.

\begin{table*}[h]
\caption{}
\adjustbox{max width=1.8\paperwidth}{
\centering
\small
\begin{tabular}{lcccc}
\hline
\multirow{2}{*}{Model}   & \multicolumn{2}{c}{Classification Accuracy}  & \multicolumn{2}{c}{Location+Repair Accuracy} \\\cline{2-5}
                         & $\leq$ 250 & $\leq$ 1000 & $\leq$ 250 & $\leq$ 1000 \\\hline
Transformer               & 75.9\%         & 73.2\%   & 67.7\%   & 63.0\% \\
RNN Sandwich              & 82.5\%         & 81.9\%   & 75.8\%   & 73.8\% \\\hline\hline
\textbf{Improvement}      &  \textbf{6.6\%}&  \textbf{8.7\%}   &  \textbf{8.1\%}   & \textbf{10.8\% } \\\hline
\end{tabular}
}
\label{tab:222}
\end{table*}

\noindent
\textit{\textbf{Method Name Prediction}} code2seq~\cite{alon2018code2seq} set the bar for large-scale, cross-project method name prediction. Table~\ref{tab:333} shows how it compares against a simple 2-layer bidirectional LSTM on code2seq's datasets: Java-small, Java-med, and Java-large.

\begin{table*}[h]
\caption{}
\adjustbox{max width=.8\paperwidth}{
\centering
\small
\begin{tabular}{lccccccccc}
\hline
\multirow{2}{*}{Model}   & \multicolumn{3}{c}{Java-small}  & \multicolumn{3}{c}{Java-med} & \multicolumn{3}{c}{Java-large} \\\cline{2-10}
                         & Precision & Recall  & F1      & Precision & Recall  & F1      & Precision & Recall  & F1 \\\hline
2-layer bi-LSTM          & 42.63\%   & 29.97\% & 35.20\% & 55.15\%   & 41.75\% & 47.52\% &63.53\%    & 48.77\% & 55.18\% \\
code2seq                 & 50.64\%   & 37.40\% & 43.02\% & 61.24\%   & 47.07\% & 53.23\% &64.03\%    & 55.02\% & 59.19\% \\\hline\hline
\textbf{Improvement}     & \textbf{8.01\%} & \textbf{7.43\%}   & \textbf{7.82\%}    & \textbf{6.09\%} & \textbf{5.32\%} & \textbf{5.71\%} & \textbf{0.50\%} & \textbf{6.25\%} & \textbf{4.01\%} \\\hline
\end{tabular}
}
\label{tab:333}
\end{table*}

Table~\ref{tab:444} shows how Sequence GNN~\cite{fernandes2018structured} compares against code2seq on Java-small. 

\citet{Wang101145} proposed a blended model, called Liger, that combines both static and dynamic program features for predicting method names. Table~\ref{tab:555} shows how Liger compares with code2seq on a subset of Java-med and Java-large.

In conclusion, the improvement the GINN-based model made over the GNN-based model in each PL task presented in this paper is comparable to that each aforementioned model made over the prior state-of-the-art. More importantly, GINN-based models only work with the datasets on which GNN-based models achieved state-of-the-art results and they are measured using the metrics proposed by the GNN-based models.

\begin{table*}[t]
\caption{}
\adjustbox{max width=.8\paperwidth}{
\centering
\small
\begin{tabular}{lccc}
\hline
\multirow{2}{*}{Model}   & \multicolumn{3}{c}{Java-small}  \\\cline{2-4}
                         & F1        & ROGUE-2  & ROGUE-L  \\\hline
code2seq                 & 43.0\%    & -        & -        \\
Sequence GNN             & 51.4\%    & 25.0     & 50.0     \\\hline\hline
\textbf{Improvement}     & \textbf{8.4\%} & -   & -        \\\hline
\end{tabular}
}
\label{tab:444}
\end{table*}

\begin{table*}[t]
\caption{}
\adjustbox{max width=.8\paperwidth}{
\centering
\small
\begin{tabular}{lcccccc}
\hline
\multirow{2}{*}{Model}   & \multicolumn{3}{c}{Java-med} & \multicolumn{3}{c}{Java-large}  \\\cline{2-7}
                         & Precision & Recall  & F1      & Precision & Recall  & F1       \\\hline
code2seq                 & 32.95\%   & 20.23\% & 25.07\% & 36.49\%   & 22.51\% & 27.84\%  \\
Liger                    & 39.88\%   & 27.14\% & 32.30\% & 43.28\%   & 31.43\% & 36.42\%  \\\hline\hline
\textbf{Improvement}     & \textbf{6.93\%} & \textbf{6.91\%}   & \textbf{7.23\%}    & \textbf{6.79\%} & \textbf{8.92\%} & \textbf{8.58\%} \\\hline
\end{tabular}
}
\label{tab:555}
\end{table*}

\end{document}